\documentclass[journal]{IEEEtran}

\usepackage{amsmath,amssymb,amsfonts,amsthm}

\usepackage{cite}
\usepackage{algorithmic}
\usepackage{textcomp}

\usepackage{dsfont,eucal,bbm,bm,nicefrac} 
\usepackage{graphicx,float} 
	\graphicspath{{./figures/}}

\newtheorem{theorem}{Theorem}

\newtheorem{lemma}[theorem]{Lemma}

\newtheorem{remark}{Remark}

\newtheorem{proposition}{Proposition}
\newtheorem{definition}{Definition}

\newcommand{\norm}[1]{\ensuremath{\left\| #1\right\|}}
\newcommand{\pbra}[1]{\ensuremath{\left( #1\right)}}
\newcommand{\sbra}[1]{\ensuremath{\left[ #1\right]}}
\newcommand{\cbra}[1]{\ensuremath{\left\{ #1\right\}}}
\newcommand{\abra}[1]{\ensuremath{\left< #1\right>}}
\newcommand{\oder}[2]{\ensuremath{\frac{d #1}{d #2}}}
\newcommand{\pder}[2]{\ensuremath{\frac{\partial #1}{\partial #2}}}

\begin{document}

\title{Learning Swarm Interaction Dynamics from Density Evolution}

\author{Christos N. Mavridis, \IEEEmembership{Member, IEEE}, 
Amoolya Tirumalai, and 
John S. Baras, \IEEEmembership{Life Fellow, IEEE}
\thanks{%
This work was supported in part by the 
Defense Advanced Research Projects Agency (DARPA) under Agreement No. HR00111990027 
and by ONR grant N00014-17-1-2622.}
\thanks{The authors are with the Department of Electrical and Computer Engineering 
and the Institute for Systems Research, 
University of Maryland, College Park, MD 20742 USA 
(e-mails: mavridis@umd.edu, ast256@umd.edu, baras@umd.edu).}
}

\maketitle

\begin{abstract}
We consider the problem of understanding 
the coordinated movements of biological or artificial swarms.
In this regard, we propose a learning scheme to estimate the 
coordination laws of the interacting agents 
from observations of the swarm's density over time.
We describe the dynamics of the swarm 
based on pairwise interactions according to a Cucker-Smale flocking model,
and express the swarm's density evolution as the solution to a system of 
mean-field hydrodynamic equations.
We propose a new family of 
parametric functions to model the pairwise interactions, 
which allows for the mean-field macroscopic system of integro-differential equations
to be efficiently solved as an augmented system of PDEs.
%
%
Finally, we incorporate the augmented system in an iterative optimization 
scheme to learn the dynamics of the interacting agents 
from observations of the swarm's density evolution over time.
The results of this work can 
offer an alternative approach to study how animal flocks coordinate, create 
new control schemes for large networked systems, and
serve as a central part of defense mechanisms against adversarial drone attacks.
%
\end{abstract}


\section{Introduction}
\label{Sec:Introduction}

\IEEEPARstart{T}{he} highly coordinated movements of animal flocks 
are among the most
fascinating phenomena to be found in nature, and understanding their 
dynamics and coordination laws has been the research focus for many scientists
over the last decades
\cite{okubo1986dynamical, reynolds_Boids_1987, cucker2007emergent, giardina2008collective, ballerini2008interaction, bajec2009organized}. 

Extracting the laws of interaction 
between agents of general networked systems finds applications 
in a wide range of fields, from 
power systems and chemical reaction networks, 
to social networks and UAV swarms
\cite{giardina2008collective, ballerini2008interaction, Matei2019, cucker2007emergent, mavridis2020detection, vicsek2010}.
%
Statistical 
\cite{lu_nonparametricInferenceOfInteractionLaws_2018},
and model-based
\cite{reynolds_Boids_1987, cucker2007emergent, Matei2019, mao2019nonlocal}
learning approaches have been 
used to learn the interaction rules between agents.
%
%
%
%
There are generally two broad approaches in modeling the underlying dynamics 
of ensembles of self-organizing agents: 
the microscopic particle models, described by ordinary 
or stochastic differential equations,
and the macroscopic continuum models, described by partial differential
equations (PDEs).
Agent-based models assume behavioral rules at the individual level, 
such as velocity alignment, attraction, and repulsion
\cite{reynolds_Boids_1987, cucker2007emergent, giardina2008collective, ballerini2008interaction}, while macroscopic models, consider large number of interacting agents,
approaching the mean-field limit.
These models typically consist of hydrodynamic PDEs 
defined on macroscopic quantities, such as the swarm's density 
\cite{carrillo2010particle, shvydkoy2017eulerian, mavridis2020semi}, and
have been studied for the analysis and control of artificial swarms,
mainly in robotic applications
\cite{sinigaglia2021density,elamvazhuthi2019mean,fornasier2014mean}.

%
Particle models have been mainly used in 
numerical simulations and learning methodologies
\cite{Matei2019, chen2016switching, mao2019nonlocal}.
Recently, Mao et al. in \cite{mao2019nonlocal}
modeled the interactions with respect to a fractional differential system of equations, 
and
Matei et al. in \cite{Matei2019} proposed an energy-based approach 
by modeling the network as a port-Hamiltonian system \cite{h2}.
However, useful real-life data of particle trajectories are difficult to extract 
and may require substantial memory and computation resources 
\cite{ballerini2008interaction,vicsek2010}. 
The experimental measurements, which usually involve digital imaging
or high-resolution GPS devices,
are difficult to acquire and are subject to artificially created noise
originating from both the sensors and the processing algorithms.
In \cite{ballerini2008interaction}, for example, stereometric and 
computer vision techniques have been used to measure long-time and long-distance
3D position trajectories of starling flocks, and in \cite{vicsek2010}, 
GPS devices were installed to homing pigeons flying in small flocks of no more than  
13 individuals.
%

On the other hand, 
useful approximations of the ensemble's density evolution 
can be easier to extract, often by applying 
simple morphological operators on vision-based recordings.
For this reason, we believe that developing learning algorithms based on the
macroscopic quantities can play a crucial role 
in the analysis of collective motion, 
and only remains inhibited due to computational expense;
%
the flocking dynamics 
can be non-local as well as nonlinear \cite{shvydkoy2017eulerian},
which results in a costly computation of 
the solution of the corresponding hydrodynamic equations
\cite{mao2019nonlocal, mavridis2020learning}.

\textbf{Contribution.} 
In this work, we introduce a modified Cucker-Smale model of 
non-local particle interaction for velocity consensus
\cite{cucker2007emergent, ha2009simple}
to efficiently solve the macroscopic hydrodynamic equations.
We propose a family of parametric interaction functions
which are shown to correspond to Green's functions associated with an
appropriately defined 
differential operator.
This allows for the transformation of the  
macroscopic hydrodynamic integro-differential equations into an 
augmented system of PDEs,
which, in turn, results in a speed-up in the computation of the non-local interaction terms.
We investigate the conditions under which time-asymptotic flocking 
is achieved, and 
%
utilize the computational advantages of the proposed methodology 
to construct an iterative optimization algorithm to learn the interaction function 
based on observations of the particle density evolution.  
Finally, we also investigate the advantages of incorporating the proposed interaction 
function model in learning algorithms based on particle trajectories (microscopic models).
%
%
The results of this work can 
be used to model and understand biological and artificial flocks 
with applications in the
control of large networked systems and artificial robotic swarms,
and in defensive mechanisms against adversarial swarm attacks.


\section{Mathematical Models and Notation}
\label{Sec:MathematicalModels}

In this section we introduce the notation that will be followed 
throughout the manuscript, define time-asymptotic flocking and
the Cucker-Smale particle dynamics, and
derive the mean-field macroscopic equations.

\subsection{The Cucker-Smale Model}

Consider an interacting system $\mathcal{G}$ of $N$ identical particles (representing
autonomous agents)
with unit mass in $\mathbb{R}^d$, $d \in \{1,2,3\}$. 
Let $x_i(t),\ v_i(t)\in\mathbb{R}^d$ represent the position and velocity 
of the $i^{th}$-particle at each time $t\geq 0$, respectively, for $1\leq i\leq N$.
Then the general Cucker-Smale system \cite{cucker2007emergent} 
is a dynamical system of $(2Nd)$ ODEs:
\begin{equation}
\begin{cases}
\oder{x_i}{t} &= v_i \\
\oder{v_i}{t} &= \frac{1}{N}\sum_{j=1}^{N}\psi(x_j,x_i)(v_j-v_i) 
\end{cases}
\label{eq:cs}
\end{equation}
where $x_i(0)$, and $v_i(0)$ are given for all $i = 1,\ldots,N$, 
and ${\psi:\mathbb{R}^d \times \mathbb{R}^d \rightarrow \mathbb{R}}$ 
represents the interaction function between each pair of particles.
%
We define the center of mass system $(x_c,v_c)$ of 
$\mathcal G = \cbra{(x_i,v_i)}_{i=1}^N$ as
\begin{equation}
x_c = \frac 1 N \sum_{i=1}^N x_i, \quad v_c = \frac 1 N \sum_{i=1}^N v_i
\end{equation}
We are interested in symmetric interaction functions $\psi(x,s)=\psi(s,x)$, 
in which case system (\ref{eq:cs}) implies
\begin{equation}
\oder{x_c}{t}=v_c,\quad \oder{v_c}{t}=0
\end{equation}
which yields a unique solution
\begin{equation}
x_c(t) = x_c(0) + t v_c(0),\ t\geq 0
\label{eq:center_of_mass}
\end{equation}

Under additional assumptions on $\psi$
(see Section \ref{sSec:Flocking}), system (\ref{eq:cs}) 
can be shown to converge to a velocity 
consensus, while preserving spatial coherence, 
a property that is known as time-asymptotic flocking, defined as follows:
\begin{definition}[Time-Asymptotic Flocking]
An $N-$body interacting system $\mathcal{G}=\cbra{(x_i,v_i)}_{i=1}^N$ 
exhibits time-asymptotic flocking with bounded fluctuation if and only if 
the following two relations hold:
\begin{itemize}
\item(Velocity alignment):  The velocity fluctuations approach zero asymptotically, i.e.
$$\lim_{t\rightarrow\infty}\sum_{i=1}^N \norm{v_i(t)-v_c(t)}^2=0$$ 
\item(Spatial coherence): The position fluctuations are uniformly bounded, i.e.
for some $0<\Lambda<\infty$,
$$\sup_{t \geq 0} \norm{x_i(t)-x_c(t)} < \Lambda,\ \forall i\in\cbra{1,\ldots,N}$$
\label{def:flocking}
\end{itemize}
\end{definition}

Throughout this article, we will be investigating flocking behaviors
and will be working with the fluctuation variables around the center of mass system, 
defined as:
\begin{equation}
(\hat x_i, \hat v_i):=(x_i-x_c, v_i-v_c)
\label{eq:fluctuations}
\end{equation}
which can be shown to satisfy the same Cucker-Smale 
dynamics described in (\ref{eq:cs}). 
We will take advantage of the spatial coherence of the flocking behavior, 
and define the position variables $\hat x_i$ 
in a compact support $D:= \cbra{x\in\mathbb{R}^d:\|x\| < L/2}$ 
for some finite $L>0$ and for all $i\in\cbra{1,\ldots,N}$,
with $\|\cdot \|$ representing the standard Euclidean norm in 
$\mathbb{R}^d$.

The set $D$ is time-dependent and 
represents a subset of $\mathbb{R}^d$ centered at the center of mass of the swarm
$x_c(t)$, $t\geq 0$, outside of which, the density of the swarm is considered negligible.
We note that time-dependent transformation (\ref{eq:fluctuations})
only requires the knowledge of the initial conditions 
$x_i(0)$ and $v_i(0)$, $i = 1,\ldots,N$.

\subsection{The Mean-Field Limit}

When the number of agents $N$ becomes large, the use of continuum models for the evolution of a density of individuals becomes essential. 
In the following, we introduce a continuum model based on the 
hydrodynamic description derived by studying the mean-field particle limit
following the Cucker-Smale model (\ref{eq:cs}).

Consider the joint probability triple of the entire particle system $\{\Omega:=\mathbb R^{2Nd}, \mathcal B(\Omega), P_{xv} \}$, the state space for each particle $\{\mathbb R^{2d}, \mathcal B(\mathbb R^{2d}) \}$ and define the empirical (random) probability measure $F_{xv}^N:\Omega \times [t_0,t_f] \times \mathcal B(\mathbb R^{2d}) \rightarrow [0,1]$ such that
\begin{equation}
    F_{xv}^N(t,A):=\frac{1}{N}\sum_{i=1}^N \mathbb I_{A}((x_i(t),v_i(t)))
\end{equation}
where $\mathbb I_{A}(\cdot)$ is the indicator function, $A \in \mathcal B(\mathbb R^{2d})$. Some authors use Dirac measures (not the Dirac delta function) in this definition. $F_{xv}^N$ is a random measure which is purely atomic.
%
Using arguments originally due to McKean and Vlasov \cite{lancellotti2005vlasov, golse2003mean}, it can be shown that there exists a deterministic and continuous $F^*_{xv}$ such that $F_{xv}^N \overset{a.e.}{\to} F^*_{xv}$ in the weak sense, and, using Ito's lemma, that the joint probability density 
 $f^*_{xv}:[t_0,t_f] \times \mathbb R^{2d} \rightarrow \mathbb R_0^+$ associated with this measure, evolves according to the forward Kolmogorov equation on $[t_0, t_f] \times \mathbb R^{2d}$: 
\begin{equation}
\begin{cases}
\partial_t f^*_{xv} + \nabla_x\cdot (vf^*_{xv})+\nabla_v\cdot (\mathcal F f^*_{xv}) 
    = 0\\
\mathcal F(t,x,v) := \int_{\mathbb{R}^{2d}}\psi(x,s)(w-v)f^*_{xv}(t,s,w)dsdw.
\end{cases}
\label{fpke}
\end{equation}

We define the marginal probability density $\rho:[t_0,t_f]\times D \rightarrow \mathbb R^+_0$ (henceforth referred to only as density) as
\begin{equation}
    \rho(t,x):= \int_{\mathbb R^d} f^*_{xv}(t,x,v)dv
\end{equation}
and the momentum density $m:[t_0,t_f]\times D \rightarrow \mathbb R^d$ and bulk velocity
$u:[t_0,t_f]\times D \rightarrow \mathbb R^d$ as
\begin{equation}
    m(t,x):= \int_{\mathbb R^d} vf^*_{xv}(t,x,v)dv:=\rho(t,x)u(t,x)
\end{equation}
where $D\subseteq \mathbb R^d$. It is additionally assumed that $\rho, m, u$ are compactly supported.
Substituting in (\ref{fpke}), we obtain the $(d+1)$ compressible Euler equations on $[t_0, t_f] \times D$ (see also \cite{carrillo2010particle}):
\begin{equation}
\begin{cases}
\partial_t{\rho} + \nabla_x \cdot m = 0 \\
\partial_t{m} + \nabla_x \cdot (\rho^{-1}mm^T) = 
\rho \mathcal L_\psi m - m \mathcal L_\psi \rho
\end{cases}
\label{eq:euler}
\end{equation}
where 
\begin{equation}
\mathcal L_\psi \phi  (t,x)= \int_{D} \psi(x,s)\phi(t,s) ds.
\label{eq:convolution}
\end{equation}
is an integral transform with kernel $\psi:D \times D \rightarrow \mathbb R$.

\section{Screened Poisson Mediated Flocking}
\label{Sec:SemiLinearPoisson}

The integral transforms in the right hand side of (\ref{eq:euler}), 
which originate from 
the non-local interaction terms in the Cucker-Smale model,
make the compressible Euler equations (\ref{eq:euler}) a system of  
partial integro-differential equations, which is extremely challenging to solve.
We approach the solution of system (\ref{eq:euler}) 
by transforming it into an augmented system of PDEs, 
in order to use 
existing numerical methods to solve it.

Suppose that, by construction, 
the interaction function $\psi$ is a Green's function associated with some linear partial differential operator 
$\mathcal L_x:L^2(D;\mathbb R) \rightarrow L^2(D;\mathbb R)$,
such that
\begin{equation}
    \mathcal{L}_x y(t,x) = \phi(t,x)
\end{equation}
implies
\begin{equation}
    y(t,x) = \int_{D} \psi(x,s) \phi(t,s) ds.
    \label{eq:greens}
\end{equation}
Then, system (\ref{eq:euler}) 
is equivalent with the augmented system 
of ($2d+2$) partial differential equations:
\begin{equation}
\begin{cases}
\partial_t{\rho} + \nabla_x \cdot m = 0 \\
\mathcal{L}_x z = \rho \\
\mathcal{L}_x y = m \\
\partial_t{m} + \nabla_x \cdot (\rho^{-1}mm^T ) = 
\rho y - z m. 
\end{cases}
\label{eq:pdes}
\end{equation}
For the global regularity of 
system (\ref{eq:pdes}) 
one can refer to \cite{shvydkoy2017eulerian} and the references therein.
A classical example for $\mathcal{L}_x$ is the operator 
associated with the Poisson equation that arises in self-gravitational hydrodynamics \cite{Truelove_1998}. 
However, in order to alleviate the computational bottleneck introduced by the non-local 
integral terms in (\ref{eq:euler}), 
the operator $\mathcal{L}_x$ needs to be defined in a way such that:
\begin{itemize}
\item[$(a)$]  the newly introduced subsystem 
\begin{equation*}
\begin{aligned}
\mathcal{L}_x z = \rho \\
\mathcal{L}_x y = m 
\end{aligned}
\end{equation*}
can be efficiently solved with numerical methods, which is the case, for example, 
if $\mathcal{L}_x$ is an elliptic operator, 
\item[$(b)$]  the Green's function $\psi$ defined in (\ref{eq:greens}) 
retain the necessary properties of an interaction function that can drive 
the Cucker-Smale model (\ref{eq:cs}) to asymptotic flocking behavior, and
\item[$(c)$]  $\mathcal{L}_x$, and consequently $\psi$, depend on a set of parameters 
that make $\psi$ appropriate to model different interaction function profiles, 
depending on the behavior of the swarm.
\end{itemize}
With this in mind, we propose 
$\mathcal{L}_x$ to be the parametrized screened Poisson 
partial differential operator 

\begin{equation}
\mathcal L_x := -\frac{1}{2k}(\partial_x^2 - \lambda^2)
\label{eq:Lx}
\end{equation}
defined in the 
domain $D:= \cbra{x\in\mathbb{R}^d:\|x\| < L/2}$
with homogeneous Dirichlet boundary conditions.
\begin{remark}
We note that the choice of the proposed operator $\mathcal{L}_x$ in (\ref{eq:Lx}) 
is not necessarily unique.
However, to our knowledge, there is no formal method to 
construct an operator $\mathcal{L}_x$, and its associated Green's function $\psi$, 
that satisfy the conditions (a), (b), and (c) as described above.
\end{remark}

\noindent
To highlight the importance of 
conditions (a), (b), and (c), 
we stress that they allow for 
the system of partial integro-differential equations (\ref{eq:euler}) to be solved 
faster, as an augmented system of PDEs.
This is in contrast to the use of a standard kernel, e.g., the fractional Laplacian used in [11], 
that results in solving a system of 
fractional partial integro-differential equations.
In the rest of this section, we will present an analysis
of the proposed family of Green's functions as 
interaction functions of a Cucker-Smale model (\ref{eq:cs}), 
in the one-dimensional case ($d=1$), which, in 
Section \ref{Sec:HigherD} will be generalized 
to higher dimensions.
When $d=1$, system (\ref{eq:pdes}) can be compactly written as
\begin{equation}
\begin{cases}
\partial_t U  + \partial_x F(U) = S(Y, U) \\
\mathcal L_x Y = U
\end{cases}
\label{eq:euler1d}
\end{equation}
where $U:=[\rho, m]^T$, $F:=[m,m^2\rho^{-1}]^T$, $S:=[0,\rho y - zm]^T$, and $Y:=[z,y]^T$.
The Green's function $\psi$ associated with the BVP introduced in (\ref{eq:Lx}) 
can be analytically computed as 
(see Appendix \ref{App:psi}):
\begin{equation}
\psi(x,s) = 
\begin{cases}
K \sigma_p(s) \sigma_m(x) & s\leq x \\
K \sigma_m(s) \sigma_p(x) & s>x 
\end{cases}
\label{eq:psi}
\end{equation}
where 
%

%
%
\begin{equation}
\begin{aligned}
K &=  -\frac k \lambda \frac{1}{e^{\lambda L} - e^{-\lambda L}} \\
\sigma_p(z) &=  2\sinh \pbra{\lambda(z+L/2)} \\
\sigma_m(z) &=  2\sinh \pbra{\lambda(z-L/2)} \\
\end{aligned}
\label{eq:psi_coef}
\end{equation}

One of the parameters of the interaction function $\psi$ in (\ref{eq:psi}),
which affects 
the flocking behavior of the system $\mathcal{G}$, 
is the size $L$ of the bounded domain $D$ in which it is defined.
The effect of the boundedness of the domain is
illustrated in Fig. \ref{fig:psi}, where, 
for different fixed values of $x$,
$\psi(x,s)$ is compared to the function 
\begin{equation}
  \hat \psi(x,s) = \frac{k}{\lambda}e^{-\lambda\|x-s\|}  
\end{equation}
which is the Green's function corresponding to $\mathcal{L}_x$ defined 
in an infinite domain.
We can interpret this effect as a tendency to avoid
the spread of the swarm in large distances with respect to the swarm's
center of mass at each time step. 

\begin{figure}[h]
        \centering
        \includegraphics[trim=33 0 55 35, clip,width=.24\textwidth]{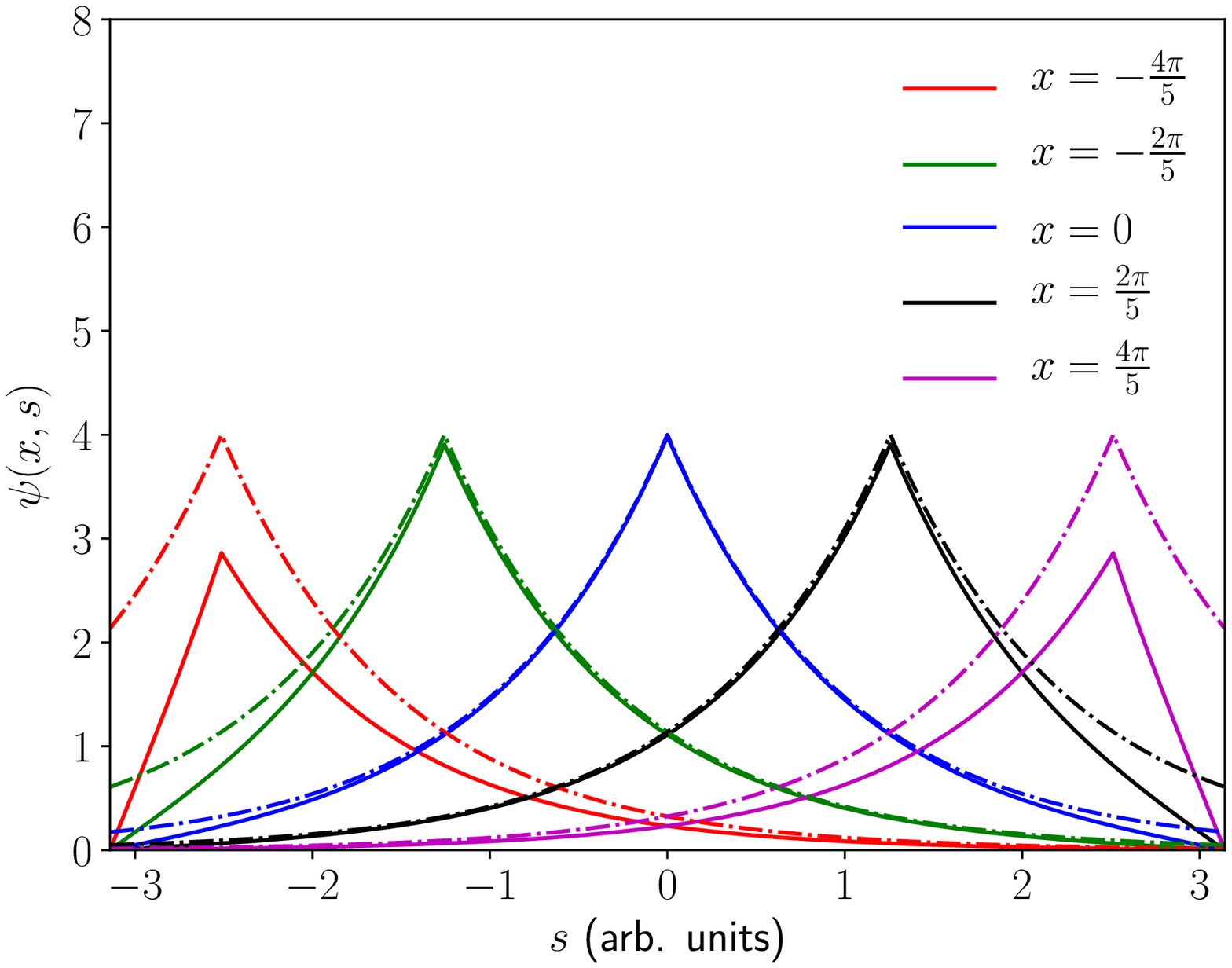}
         \includegraphics[trim=33 0 55 35, clip,width=.24\textwidth]{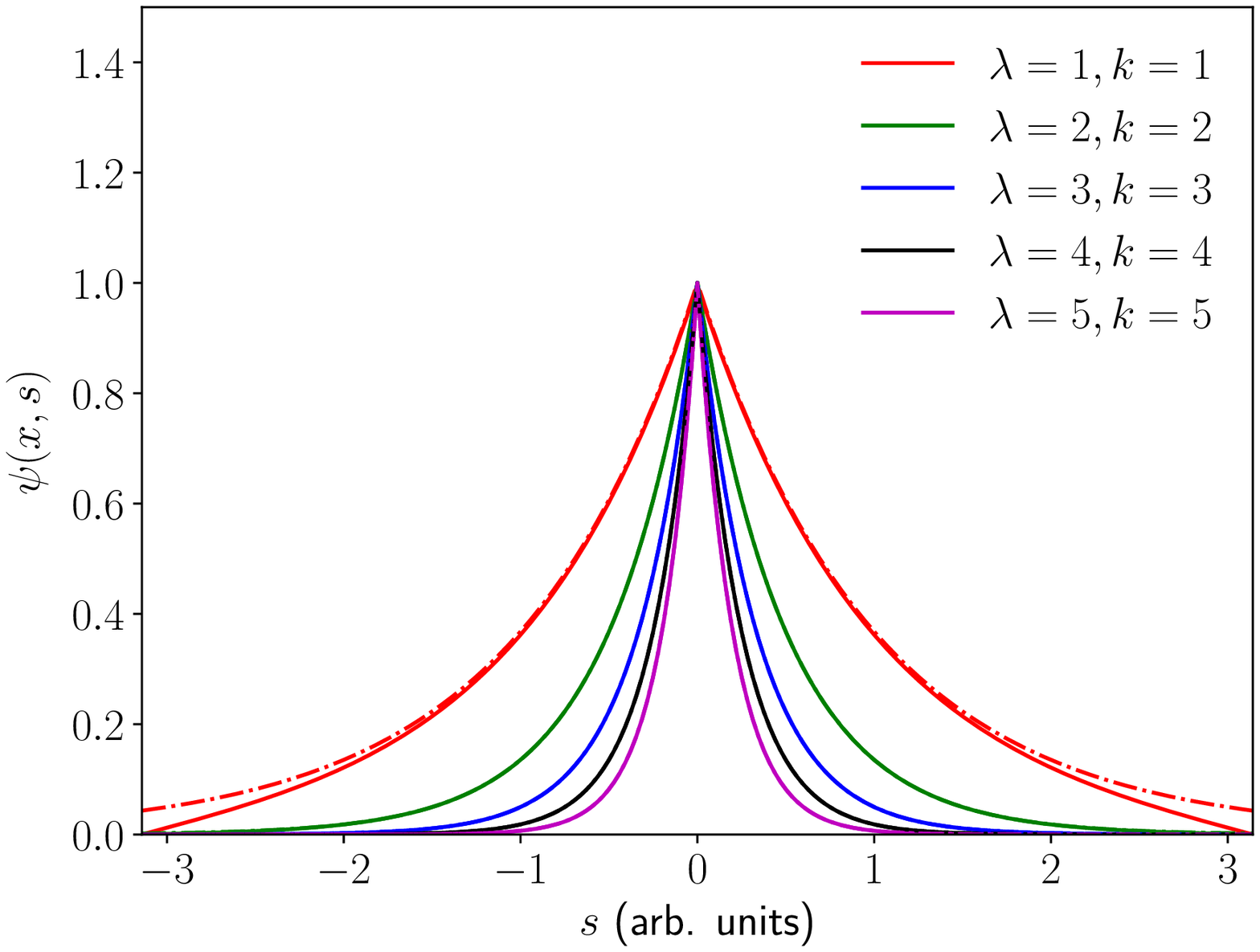}
        \caption{(left) Illustration of $\psi(x,\cdot)$ (\ref{eq:psi}) 
        for different values of $x$, and for 
        $\lambda=1$, $k=4$ on $[-\pi, \pi]$. 
        The function $\hat \psi(x,s)= \frac{k}{\lambda}e^{-\lambda\|x-s\|}$, 
        is depicted in dashed lines.
        (right) The effect of the parameters $k$, and $\lambda$ on the 
        profile of the interaction function $\psi$. The function 
        $\hat \psi$
        is depicted in dashed lines. Here, $x=0$.
        }
        
        \label{fig:psi}
\end{figure}

The parameters $k$ and $\lambda$ control the 
profile of the interaction function $\psi$ 
by affecting the influence factor of each agent to its neighborhood,
essentially changing the communication radius of each agent.
This effect is similar to the parameters $K, \gamma$ in the original 
proposed interaction function  
${\psi_{CS}(x,s) = \frac{K}{(1+\|x-s\|^2)^\gamma}}$
from Cucker and Smale
\cite{cucker2007emergent}.
As a result, a wide range of flocking behaviors can be modeled 
using the Cucker-Smale model and the 
proposed parametric interaction function $\psi$.

\subsection{Asymptotic Flocking}
\label{sSec:Flocking}

In this section we investigate the sufficient conditions 
on the initial conditions $x_i(0)$,
$v_i(0)$, $i=1,\ldots,N$, and 
the size $L$ of the domain,
such that, under the new interaction function (\ref{eq:psi}), (\ref{eq:psi_coef}), 
the solution $\cbra{(x_i(t), v_i(t))}_{i=1}^N$, $t\geq 0$, 
of system (\ref{eq:cs}) satisfies the flocking conditions of 
Definition \ref{def:flocking}.

We define
$\hat x:=(\hat x_1,\ldots,\hat x_N)$, $\hat v:=(\hat v_1,\ldots,\hat v_N)$ $\in\mathbb{R}^{Nd}$,
$|\hat x| = \pbra{\frac 1 N \sum_{i=1}^N \|\hat x_i\|^2}^{1/2}$, and 
$|\hat v| = \pbra{\frac 1 N \sum_{i=1}^N \|\hat v_i\|^2}^{1/2}$,
where $\|\cdot\|$ represents the standard Euclidean norm 
in $\mathbb{R}^d$. 
Notice that $|\cdot|$ is equivalent to the
Euclidean norm in $\mathbb{R}^{Nd}$, which we associate with the inner product
$\abra{\cdot,\cdot}$ such that $\abra{x,x} = |x|^2$.
Then the following theorem holds:

\begin{theorem}

If 
$\frac 1 2 \max_{1\leq i,j\leq N}\|\hat x_i(0)-\hat x_j(0)\|< \hat x_M $
for some $\hat x_M < \frac L 4 $, 
where $L$ defines the domain $D$ in (\ref{eq:Lx}), and
$$ |\hat v_0| < \int_{|\hat x_0|}^{\hat x_M} \psi(-2\hat x_M,\lambda s) ds, $$
for some $\lambda>0$,
then the solution $\cbra{(x_i(t), v_i(t))}_{i=1}^N$, $t\geq 0$, 
of system (\ref{eq:cs}) satisfies the flocking conditions of 
Definition \ref{def:flocking}.
\label{thm:flocking1D}
\end{theorem}

\begin{proof}
The result follows by showing that the solution $(|\hat x(t)|,|\hat v(t)|)$ 
satisfies the system of dissipative differential inequalities: 
\begin{equation*}
\oder{|\hat x|}{t} \leq |\hat v|,\quad 
\oder{|\hat v|}{t} \leq -\phi(|\hat x|)|\hat v|
\end{equation*}
We first notice that, from the Cauchy-Schwartz inequality: 
\begin{equation}
\oder{|\hat x|^2}{t}=\frac 2 N \abra{\oder{\hat x}{t},\hat x}
=\frac 2 N \abra{\hat v,\hat x}\leq 2|\hat x||\hat v|
\end{equation}
Because $\oder{|\hat x|^2}{t} = 2 |\hat x| \oder{|\hat x|}{t}$, this implies that
\begin{equation}
\oder{|\hat x|}{t} \leq |\hat v|.
\label{eq:bound_x}
\end{equation}
%
%
Now we have assumed that for the initial conditions $\cbra{\hat x_i(0)}$, 
$L$, which is a design parameter, is large enough such that there exist an 
$\hat x_M\in[0,\frac L 4)$
for which 

\begin{equation}
|\hat x(0)| 
< \hat x_M
\end{equation} 
since $|\hat x| \leq \nicefrac 1 2 \max_{1\leq i,j\leq N}\|\hat x_i-\hat x_j\|$ \cite{shiffler1980upper}. 
From the definition of the $\psi$ function in 
(\ref{eq:psi}), (\ref{eq:psi_coef}), 
it follows that for $|\hat x| \leq \hat x_M$, 
\begin{equation}
\begin{aligned}
\psi(x_j,x_i) &\geq \psi(-2\hat x_M, \|\hat x_j-\hat x_i\|) \\
&\geq \psi(-2\hat x_M, \lambda|\hat x|)
\end{aligned}
\end{equation}
for some $\lambda>0$ 
such that 
$\lambda | \hat x | \geq \max_{1\leq i,j\leq N}\|\hat x_i-\hat x_j\|$.
This implies that
\begin{equation}
\begin{aligned}
\oder{|\hat v|^2}{t} &=-\frac 1 {N^2} \sum_{1\leq i,j\leq N} \psi(\hat x_j, \hat x_i)\|\hat v_j-\hat v_i\|^2 \\
&\leq -\frac 1 {N^2} \psi(-2\hat x_M, \lambda|\hat x|) \sum_{1\leq i,j\leq N} \|\hat v_j-\hat v_i\|^2 \\
&\overset{(*)}{=} - 2 \psi(-2\hat x_M, \lambda|\hat x|) |\hat v|^2
\end{aligned}
\end{equation}
and, consequently,

\begin{equation}
\oder{|\hat v|}{t} \leq - \psi(-2\hat x_M, \lambda|\hat x|) |\hat v| :=-\phi(|\hat x|)|\hat v|
\label{eq:bound_v}
\end{equation}
In step $(*)$, we have used the fact that 
\begin{equation*}
\begin{aligned}
\sum_{1\leq i,j\leq N} \|\hat v_j-\hat v_i\|^2 &= 
2N \sum_{i = 1}^{N} \|\hat v_i\|^2 
- 2 \abra{\sum_{i = 1}^{N} \hat v_i, \sum_{j = 1}^{N} \hat v_j} \\
&= 2{N^2} |\hat v|^2
\end{aligned}
\end{equation*}
since $ \sum_{i=1}^N \hat v_i(t)=0 $, $t\geq 0$. 
Next we notice that the Lyapunov function
\begin{equation}
V(|x|,|v|):=|\hat v|+ \int_{\alpha}^{|\hat x|} \phi(s) ds,\ \alpha\geq 0 
\end{equation}
is non-increasing along the solutions of $(|\hat x(t)|,|\hat v(t)|)$ 
of the system of dissipative differential inequalities 
(\ref{eq:bound_x}) and (\ref{eq:bound_v}), 
since
\begin{equation}
\begin{aligned}
\oder{}{t}V(|\hat x|,|\hat v|) & = \oder{|\hat v|}{t} + \phi(|\hat x|)\oder{|\hat x|}{t} \\
&\leq \phi(|\hat x|) 
\pbra{ -|v| + \oder{|\hat x|}{t}} \\
&\leq 0
\end{aligned}
\end{equation}
which implies that 
\begin{equation}
|\hat v(t)|+ \int_{|\hat x_0|}^{|\hat x|} \phi(s) ds 
\leq |\hat v(0)|
\label{eq:contradiction1}
\end{equation}
and 
\begin{equation}
|\hat x| \leq \hat x_M
\end{equation}
as long as $|\hat x_0| \leq \hat x_M$. 
This means that $\max_{1\leq i,j\leq N}\|\hat x_i-\hat x_j\| \leq \lambda \hat x_M$ 
and the spatial coherence requirement
of Definition \ref{def:flocking} is satisfied for some 
$\Lambda >0$.

Regarding the velocity consensus, 
we have assumed that the initial velocity $|\hat v(0)|$ 
satisfies 
\begin{equation}
|\hat v(0)|<\int_{|\hat x(0)|}^{\hat x_M} \phi(s)ds
\end{equation}
and, since $\phi$ is non-negative for $|\hat x(t)| \leq \hat x_M$, 
there exists a $\bar x \in [|\hat x(0)|,\hat x_M]$ for which 
\begin{equation}
|\hat v(0)|=\int_{|\hat x(0)|}^{\bar x} \phi(s)ds 
\end{equation} 
Suppose there exists a $t^*\geq 0$, such that $\hat x^*:=|\hat x(t^*)|\in (\bar x,\hat x_M]$.
Then 
\begin{equation}
\int_{|\hat x(0)|}^{\hat x^*} \phi(s)ds > |v(0)| 
\end{equation}
which contradicts (\ref{eq:contradiction1}).
Therefore, 
%
from (\ref{eq:bound_v}) and the Gr{\"o}nwall-Bellman inequality 
\begin{equation}
|\hat v(t)| \leq |\hat v(0)|e^{-\phi(\bar x)t},\ t\geq 0
\end{equation}
i.e., the flocking conditions of Definition 
\ref{def:flocking} are satisfied.
\end{proof}

We note that if the conditions of Theorem \ref{thm:flocking1D} do not hold, then 
flocking is possible but not guaranteed. In \cite{ha2009simple}, similar 
conditions and their effect on the flocking behavior of the swarm are investigated.

\subsection{Conservation of Mass and Momentum}
\label{sSec:Conservation}

Next, we show that, in 
system (\ref{eq:euler1d}) with the operator $\mathcal L_x$ 
as defined in (\ref{eq:Lx}), 
mass and momentum are conserved.

\begin{lemma} The operator $\mathcal L_x$ (\ref{eq:Lx}) 
is self-adjoint
and invertible, and therefore has a self-adjoint inverse 
$\mathcal L_x^{-1}$.
\end{lemma}
\begin{proof} 
Self-adjointness of the inverse follows immediately from self-adjointness of $\mathcal L_x$ and the existence of the inverse \cite{taylor2010partial}. It is clear that $\mathcal L_x$ has an inverse since the Green's function is nontrivial as given by its sine series. Self-adjointness of $\mathcal L_x$ follows as a direct application of integration by parts and Green's second identity 
\cite{evans1998partial}.
%
\end{proof}

\begin{proposition}
If $Y\in C^{\infty}_{\mathbb R, \mathcal C}(D)$, then mass
and momentum are conserved, i.e.
\begin{equation}\label{cons2}
\frac{d}{dt}\int_D U d x = 
    \int_D S d  x = 0.
\end{equation}
\end{proposition}
\begin{proof}
We obtain (\ref{cons2}) by simply integrating the balance
laws in (\ref{eq:euler1d}) over $D$ and apply the Leibniz rule.
The conclusion follows directly from the self
adjointness of the inverse proved in Lemma 1. 
\end{proof}

\subsection{Computational Methods}
\label{sSec:ComputationalMethods}

Adopting the proposed interaction function form 
(\ref{eq:psi}), (\ref{eq:psi_coef}),  
results in the system of PDEs (\ref{eq:euler1d}). 
We describe here the computational methods used 
to efficiently solve (\ref{eq:euler1d}) and 
compute the macroscopic quantities,
i.e. the momentum and density. 

\subsubsection{Hyperbolic Solver}
\label{ssSec:HyperbolicSolver}

To solve the hyperbolic system of (\ref{eq:euler1d}), 
we apply the finite volume method \cite{leveque_2002}. 
We define the sequence of points 
$x_s = \{x_1, ..., x_i, ..., x_{N_s} \}$
which are the centers of the cells 
$I_i := [x_{i-\frac{1}{2}}, x_{i+\frac{1}{2}})$, 
and average the PDE
over these cells, which gives
\begin{equation}
    \frac{1}{\lambda(I_i)}\frac{d}{dt}\int_{I_i} U dx = 
    -\frac{1}{\lambda(I_i)}\int_{I_i} \partial_{x} F dx + 
    \frac{1}{\lambda(I_i)}\int_{I_i} S dx
\end{equation}
where $\lambda(\cdot)$ is the Lebesgue measure. 
Assuming these are
identical, such that $\Delta x :=\lambda(I_i) \forall i$, 
we can make use of the divergence
theorem, and replace the integrals of $U,F, S$ with their cell-averages,
i.e. their midpoint values $\bar U, \bar F, \bar S$, 
in order to obtain semi-discrete scheme:
\begin{equation}
    \frac{d}{dt} \bar U_i = 
    -\frac{1}{\Delta x}(\bar F_{i+\frac{1}{2}} - \bar F_{i-\frac{1}{2}}) + 
    \bar S_i
\end{equation}
where $\bar U_i := \bar U(x_i), \bar F_i := \bar F(x_i), \bar S := \bar S(x_i)$.
For the fluxes, we assume piecewise linearity and use the Kurganov-Tadmor flux \cite{KURGANOV2000241} given by
\begin{equation}
\begin{split}
    \bar F_{i+\frac{1}{2}}&:=\frac{1}{2}[F^+_{i+\frac{1}{2}} + F^-_{i+\frac{1}{2}} - \max_\pm \{|u^\pm_{i+\frac{1}{2}}|\}(U^+_{i+\frac{1}{2}}-U^-_{i+\frac{1}{2}})]\\
    U^+_{i + \frac{1}{2}} &:= U_{i+1} - \frac{\Delta x}{2}minmod(\frac{U_{i+2} - U_{i+1}}{\Delta x},\frac{U_{i+1} - U_{i}}{\Delta x})\\
     U^-_{i + \frac{1}{2}} &:= U_{i} + \frac{\Delta x}{2}minmod(\frac{U_{i+1} - U_{i}}{\Delta x},\frac{U_{i} - U_{i-1}}{\Delta x})\\
    \end{split}
\end{equation}
where $minmod(a,b) := \frac{1}{2}(sign(a)+sign(b))\min(|a|, |b|)$. 


\subsubsection{Elliptic Solver}

To solve the elliptic equations of (\ref{eq:euler1d}), we employ spectral methods. Noting that a basis for the space of $L^2((0,L);\mathbb R)$ functions with zero BCs is the sequence $\{b_n(x):=\sin \frac{n\pi x}{L} \}_{n \in \mathbb N}$, we propose candidate solutions to the elliptic BVP for fixed $t$ as Fourier sine series:
\begin{equation}
    \phi(x,t) = \sum_{n=1}^\infty \hat \phi_n(t)b_n(x'),
\end{equation}
where $x'(x) = x + \frac{L}{2}$
Now, we apply the operator $\mathcal L_x$ to $\phi$, which yields:
\begin{equation}
    \sum_{n=1}^\infty \frac{1}{2k}(\mu_n + \lambda^2)\phi_n(t) b_n(x') = q(x',t)
\end{equation}
where $\mu_n := (\frac{n\pi}{L})^2$ are the eigenvalues of $\partial_x^2(\cdot)$ with Dirichlet BCs.
Now, let $\hat q_n(t)$ denote the $n-$th Fourier sine coefficient for $q(x,t)$. Considering an approximation to $\phi$ with $N_s$ harmonics corresponding to the same $N_s$ as in the hyperbolic solver, we obtain the semi-discrete spectral method:
\begin{equation}
    \hat \phi_n(t) = \frac{2k \hat q_n(t)}{\mu_n + \lambda^2}, \text{ } 1 \leq n \leq N_s
\end{equation}
We implement this spectral method using discrete sine transform (DST) II in the forward direction and sine transform III in the backward direction to obtain the approximation of $\phi$ from its sine coefficients. The spectral method is converted into a fully discrete scheme according to the temporal discretization of the semi-discrete scheme of the hyperbolic solver. 
%

\begin{remark}
We note that, in 1D, the computation time of using a direct convolution sum (parallelized)
to compute the integral term of the original system (\ref{eq:euler})
has complexity $O(N_s^2)$ (where $N_s$ is the number of cells), since a sum is required for each point on the line where 
the convolution is to be approximated.
In contrast, the FFT-based elliptic solver has complexity $O(N_s \log(N_s) + N_s)$,
where the added $N_s$ corresponds to multiplication of coefficients.
The difference becomes even more significant in higher dimensions, as explained in 
Section \ref{Sec:HigherD}. 
Fig. 2 presents a quantitative comparison.
\label{rmk:1D}
\end{remark}

%
%
%
%


\begin{figure}[ht!]
        \centering
        \includegraphics[trim=0 0 40 40,clip,width=0.24\textwidth]{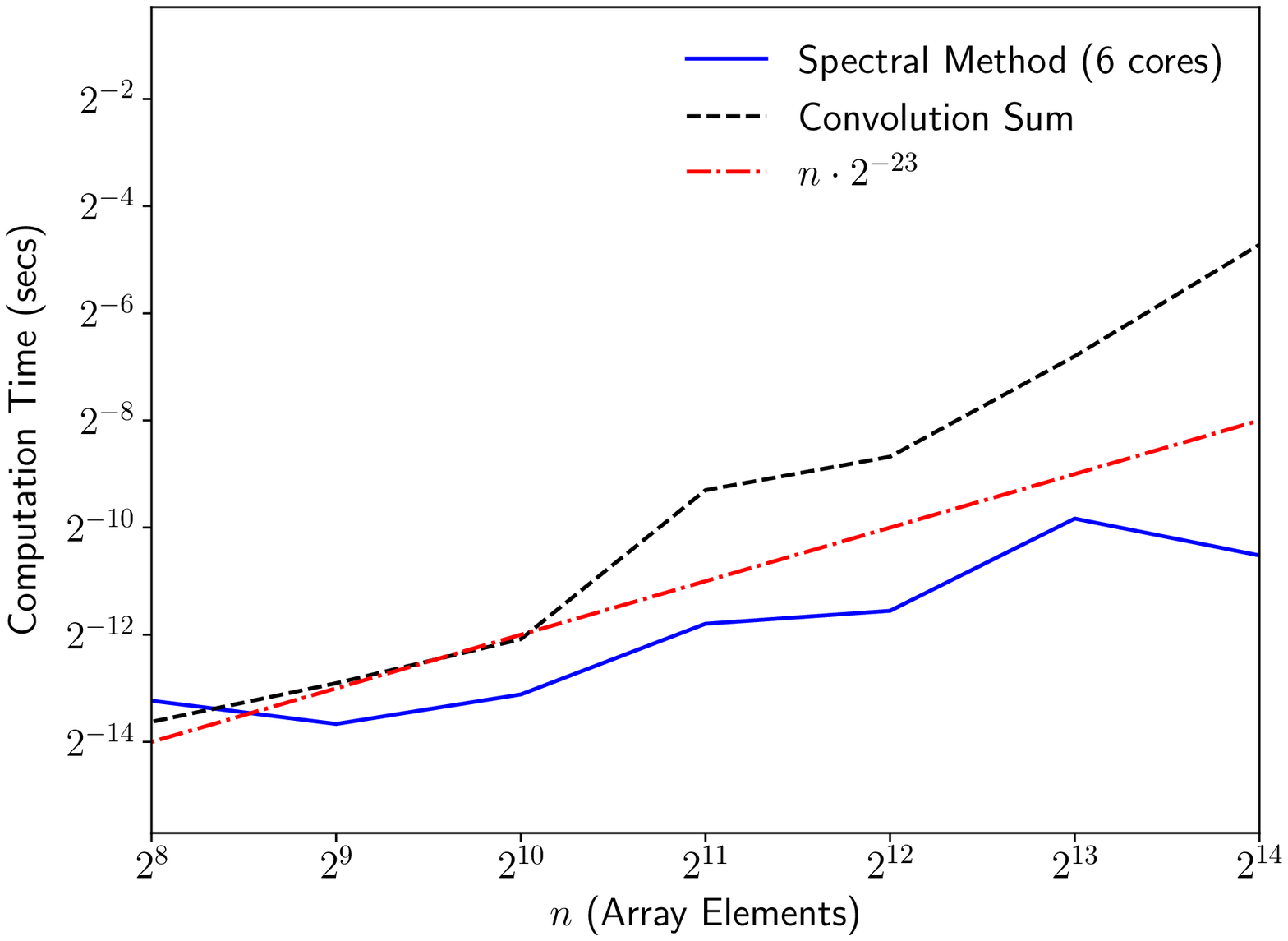}
        \includegraphics[trim=0 0 40 40,clip,width=0.24\textwidth]{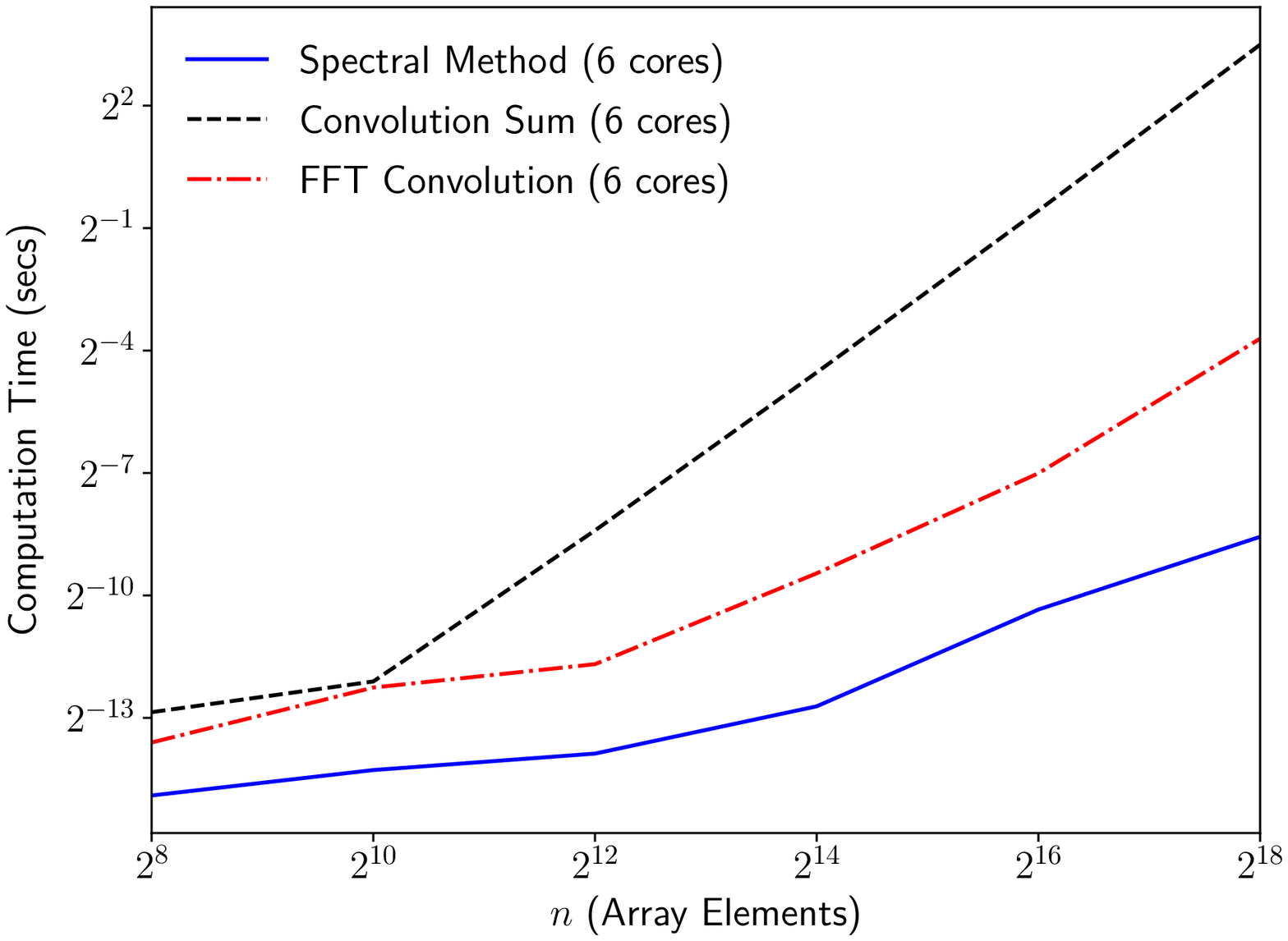}
        \caption{Computation Times for the non-local terms of (\ref{eq:euler}). (left) One dimension. 
        (right) Two dimensions. 
        The methods are comparable for very coarse grids, but spectral methods 
        rapidly become faster as more cells are added. 
        The number of cells scales quadratically with the domain size.}
        \label{fig:speed}
\end{figure}

\section{Higher Dimensions}
\label{Sec:HigherD}

The methodology outlined above is scalable and 
can be generalized to higher dimensions, as shown next.

\subsection{Screened Poisson Mediated Flocking in Radially Symmetric Domain}
It seems natural for the interaction function $\psi$ to be radially symmetric, 
which suggests that the domain $D$ has radial symmetry as well. 
In higher dimensions, i.e. for $d=2,3$, this results in singular kernels $\psi$
\cite{Ahn2012OnCI}.
Singular kernels have been extensively studied in the literature and, under
mild assumptions in the initial conditions, have been shown to result in flocking behavior
while, at the same time, avoiding collisions \cite{Ahn2012OnCI}.
In this case, we have the BVP of the augmented system of PDEs (\ref{eq:pdes})
defined in the radially symmetric domain $D:= \cbra{x\in\mathbb{R}^d:\|x\| < L/2}$, 
with the linear differential operator $\mathcal{L}_x$ defined as
\begin{equation}
\mathcal{L}_x= - k^{-d/2} ( \nabla_x^2 - \lambda^2)
\end{equation}
It can be shown (\cite{mavridis2020semi}) that 
this operator is associated with a Green's function of the form
\begin{align}
\psi(x,s) = \bar \psi(x-s) + \phi(x,s)
\end{align}
where $\bar \psi$ is given by
\begin{equation}
\begin{aligned}
\bar \psi(x,s) &= \tilde\psi(\|x-s\|) \\
&= \pbra{\frac{k}{2\pi}}^{d/2} \pbra{\frac{\lambda}{\|x-s\|}}^{d/2-1} 
K_{d/2-1}(\lambda \|x-s\|)
\end{aligned}
\end{equation}
with $K_\alpha(\cdot)$ being the modified Bessel function of the second kind of order $\alpha$, 
and 
%

\begin{equation}
\phi(x,s) = - \tilde \psi(\frac 2 L \|x\|\| s - \frac{L^2}{4} \frac{x}{\|x\|^2}\|).
\end{equation}

 

\subsection{Screened Poisson Mediated Flocking in Rectangular Domain}

The introduction of the operator $\mathcal{L}_x$ allows 
for the fast computation of the solution of (\ref{eq:euler}) by numerically 
solving 
(\ref{eq:pdes}). The hyperbolic and elliptic solvers
introduced in \ref{sSec:ComputationalMethods}, however, are computationally costly
when not working in a rectangular domain $D$. For this reason, we define 
the BVP of the augmented system of PDEs (\ref{eq:pdes})
with the same linear differential operator $\mathcal{L}_x$ defined as
\begin{equation}
\mathcal{L}_x= -\frac{1}{2k}(\nabla_x^2 - \lambda^2),~  k>0,\ , \lambda \in \mathbb R
\end{equation}
in a $d$-dimensional 
rectangular domain  $D:= \sbra{-\nicefrac{L}{2}, \nicefrac{L}{2}}^d$, $L>0$,
with homogeneous Dirichlet boundary conditions.
The intuition behind this selection is that $L$ can be chosen large enough 
to approximately negate the effect of the rectangular domain $D$ on 
the interaction function $\psi$ as shown in 
Fig. \ref{fig:psi_contours_2D}.

Notice that, as shown in Section \ref{sSec:Conservation}, 
$\mathcal{L}_x$ is an elliptic, self-adjoint (symmetric) partial differential operator 
that conserves mass and momentum. Therefore, in two-dimensions, the augmented system (\ref{eq:pdes}) takes the form:
\begin{equation}
\begin{cases}
\partial_t Q + \partial_x F(Q) + \partial_y G(Q) = S(Q, \Phi) \\
\mathcal L_x \Phi = Q
\end{cases}
\label{eq:pdes2}
\end{equation}
where $Q:=(\rho, m_1,m_2)^T$, $F:=u_1 Q$, $G:=u_2 Q$, and
$S:=(0, \rho \mathcal L_\psi m - m \mathcal L_\psi \rho)$.
System (\ref{eq:pdes2}) can be generalized to three dimensions in the obvious way.

In Fig. \ref{fig:density_contours_2D}, we illustrate the 
density and momentum density field of the solution of system (\ref{eq:pdes2})
for the initial conditions given in Section \ref{sSec:results2D}. 

%
\begin{figure}[H]
    \centering
    \includegraphics[trim=40 30 30 50,clip,width=0.5\textwidth]{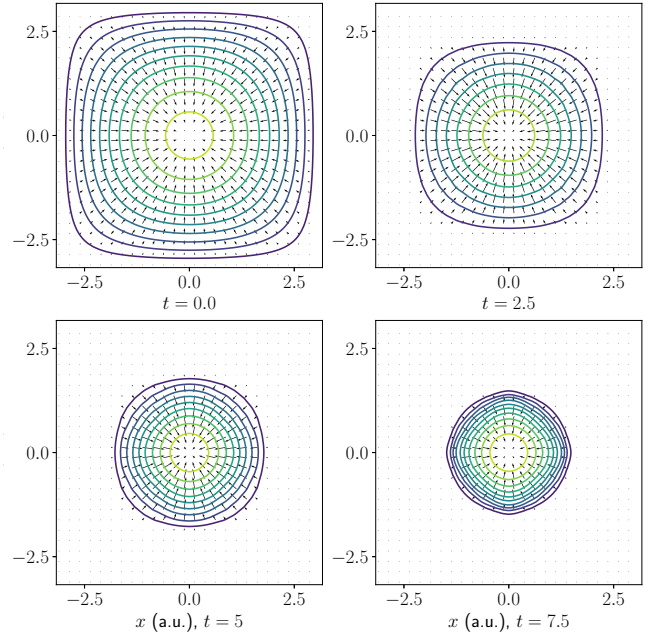}
    \caption{Density and momentum field as a solution of 
    system (\ref{eq:pdes2}) for the initial conditions given in Section \ref{sSec:results2D}. 
    The contours correspond to the density and the quivers to the momentum field. Here, $(k,\lambda) = (4,1)$.
    The timestamps of the solution are printed at the bottom of each plot.}
    \label{fig:density_contours_2D}
\end{figure}

In physics and computer graphics, this operator with $\lambda \neq 0$ is associated to the time-independent Klein-Gordon equation and the screened Poisson equation \cite{screenedPoissonEq}. In the square region $D:=(-\frac{L}{2},\frac{L}{2}) \times (-\frac{L}{2},\frac{L}{2})$ with homogeneous Dirichlet boundary conditions, the Green's function for $\mathcal L_x$ is given by the solution to
\begin{equation}
    \label{eq:slp2D}
    \begin{cases}
    \mathcal L_x \psi(x,s) = \delta(x-s) &\text{ for } (x,s) \in D \times D \\
    \psi(x,s) = 0 &\text{ for } (x,s) \in \partial D \times \partial D
    \end{cases}
\end{equation}
which is the Fourier sine series
\begin{equation}
\begin{aligned}
    \psi(x,s) &= 8k\sum_{m=1}^\infty \sum_{n=1}^{\infty} \frac{1}{\mu_{nm} + \lambda^2} \\
    &\quad\quad\quad \sin\frac{n\pi}{L}x_1'\sin\frac{n\pi}{L}s_1'\sin\frac{m\pi}{L}x_2'\sin\frac{m\pi}{L}s_2'
\end{aligned}
\end{equation}
where 
$\mu_{n,m} = (\frac{n\pi}{L})^2 + (\frac{m\pi}{L})^2$, and $s_i' = s_i + \frac{L}{2}$, and similarly for $x'$ i.e. a translation of coordinates. 
This may be easily verified by separation of variables, or simply computing $\mathcal L_x \psi$. It is obvious $\psi(x,s)$ is symmetric in its arguments, and that it is singular along $x=s$. Moreover, via Hopf's maximum principle \cite{evans1998partial, kreith1968}, it is clear immediately that $\psi(x,s)>0$ $\forall (x,s) \in D \times D$. 
So, following results in \cite{shvydkoy2017eulerian, Ahn2012OnCI}, $\psi(x,s)$ can be shown to induce flocking dynamics, as well as collision avoidance. 


Numerical approximations to the Green's function $\psi$ computed via a spectral method are presented here. The behavior of this Green's function is similar to the 1D Green's function in $k,\lambda$, although in the 1D case, the Green's function has a simple closed-form, and is nonsingular. 
In Fig. \ref{fig:psi_contours_2D}, we illustrate
the effect of the parameters $k$, $\lambda$ on the 
profile of the 2D interaction function. 
The parameter $k$ has an obvious effect on scaling, 
and $\lambda$ has
the effect of increasing its growth rate. There are singularities along $x=s$.

\begin{figure}[ht!]
        \centering
        \includegraphics[trim=25 25 30 50,clip,width=0.48\textwidth]{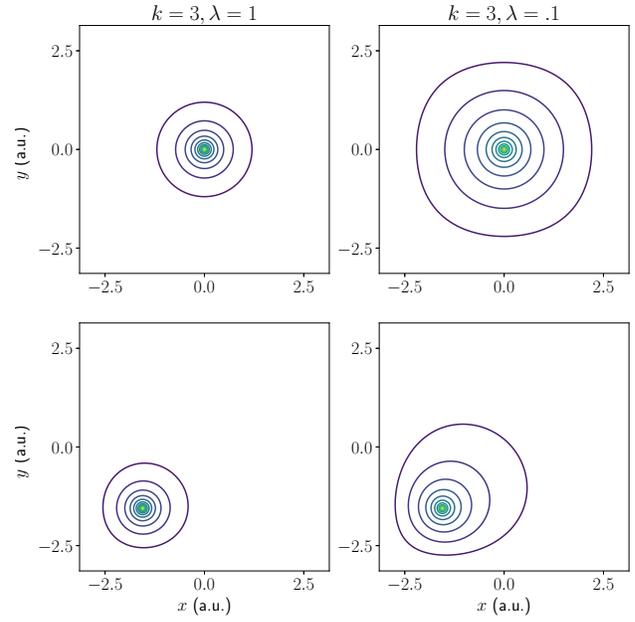}
        \caption{The effect of the parameters $k$, $\lambda$ on the 
        profile of the interaction function $\psi$ in the 
        two-dimensional rectangular domain  $D:= \sbra{-\pi, \pi}^2$. 
        The parameter $k$ has an obvious effect on scaling and  
        $\lambda$ has the effect of increasing its growth rate. 
        There are singularities along $(x,y) = (s_1,s_2)$. In the first row, 
        $(s_1,s_2) = (0,0)$ and in the second row, $(s_1,s_2) = -(\frac{\pi}{2}, \frac{\pi}{2})$.}
        \label{fig:psi_contours_2D}
\end{figure}

\subsection{Computational Methods}
\subsubsection{Hyperbolic Solver}
As in the 1D case, we apply the finite volume method \cite{leveque_2002} to convert the hyperbolic PDE system (\ref{eq:pdes}) into a system of ODEs on cells $I_{ij}$ centered on sequence of points $\{(x_i,y_j) \}_{i=1,j=1}^{N_s,N_s}$ evenly spaced with spacing $\Delta x = \Delta y$:
\begin{equation}
\begin{aligned}
        \frac{d}{dt} \bar Q_{i,j} &= 
    -\frac{1}{\Delta x}(F_{i+\frac{1}{2},j}  - F_{i-\frac{1}{2},j}) \\
    &\quad\quad\quad -\frac{1}{\Delta y}(G_{i,j+\frac{1}{2}}  - G_{i,j-\frac{1}{2}}) + 
    \bar S_{i,j}
\end{aligned}
\end{equation}
The interpolated fluxes are given by the Kurganov-Tadmor fluxes analogously to the 
1D case (\cite{KURGANOV2000241}).
%
We use the same time integration and limiting procedure as we do in 1D. The form of the TVD limiter of the Kurganov-Tadmor flux changes slightly in 2D. See \cite{KURGANOV2000241} for details.
 
\subsubsection{Elliptic Solver}
To solve the elliptic equations, we employ spectral methods, as in 1D. 
A basis for $L^2((0, L)^2;\mathbb R)$ with zero BCs is 
\begin{equation}
    \cbra{b_{n,m}(x,s):=\sin \frac{n\pi x}{L} \sin \frac{m\pi s}{L} }_{n=1,m=1}^{\infty,\infty}
\end{equation}
As in the 1D case, candidate solutions to the elliptic BVP for fixed $t$ are Fourier sine series:
\begin{equation}
    \phi(x,s,t) = \sum_{n,m=1}^\infty \hat \phi_{n,m}(t)b_{n,m}(x',s'),
\end{equation}
where $s' = s + \frac{L}{2}$ and similarly for $x'$.
Now, we apply the 2D operator $\mathcal L_x$ to $\phi$, which yields:
\begin{equation}
    \sum_{n,m=1}^\infty \frac{1}{2k}(\mu_{n,m} + \lambda^2)\hat \phi_{n,m}(t)b_{n,m}(x',s') = q(x',s',t)
\end{equation}
Now, let $\hat q_{n,m}(t)$ denote the $n,m-$th Fourier sine coefficient for $q(x',s',t)$. Considering an approximation to $\phi$ with $N_s^2$ harmonics ($N_s$ in each direction) corresponding to the same $N_s$ as in the hyperbolic solver, we obtain the semi-discrete spectral method:
\begin{equation}
    \hat \phi_{n,m}(t) = \frac{2k \hat q_{n,m}(t)}{\mu_{n,m} + \lambda^2}, \text{ } 1 \leq n \leq N_s, 1 \leq m \leq N_s
\end{equation}
where $\mu_{n,m} := (\frac{n\pi}{L})^2 + (\frac{m\pi}{L})^2$ are the eigenvalues of the Laplacian with Dirichlet BCs. We apply the multidimensional extension of the transforms used in the 1D case to implement this spectral method.

\begin{remark}
We note that, compared to Remark \ref{rmk:1D}, in higher dimensions, i.e., 2D and 3D, 
one can take advantage of the divide-and-conquer approach of FFTs as well as parallelization. 
While a direct 2D convolution sum has complexity $O(N_s^4)$, since one needs to compute a double-sum for each grid point desired, 
the FFT-based elliptic solver has complexity $O(2N_s^2 \log(N_s) + N_s^2)$. 
Please refer to Fig. 2 for quantitative results.
\label{rmk:2D}
\end{remark}




\section{Learning the Coordination Laws}
\label{Sec:density_learning}

We utilize the methodology and 
the computational methods described above to 
efficiently compute the macroscopic quantities, 
i.e. the momentum and density of the swarm, as a solution to the 
augmented system of equations (\ref{eq:pdes}). 
We now incorporate the computation of the swarm's momentum and density
in an iterative learning scheme to estimate the parameters of the 
interaction function $\psi$.


We formulate the process of learning the interaction function $\psi$ from 
density data as a PDE-constrained optimization problem:
%
%
%
\begin{equation}
    \min_{k,\lambda}
    \sum_{\tau=t_0}^{t_f} D_{KL}
    (P^*(\tau) || P(\tau))
    \label{eq:density_learning}
\end{equation}
where $P^*(t)$ and $P(t)$ 
are probability measures that have densities $\rho^*(t,\cdot)$ and $\rho(t,\cdot)$,
the observed and simulated mean-field densities, respectively. 
The density $\rho^*$ is assumed given by observation.
The mean-field density $\rho$ 
associated with $P$, is subject to the system of PDEs (\ref{eq:pdes})
and therefore dependent on the parameter vector $\theta:= (k, \lambda)$.
The Kullback-Leibler (KL) divergence $D_{KL}$ in (\ref{eq:density_learning}) is given by:
\begin{equation}
    D_{KL}(P_i||P_j):= \int_{\Omega} \log_2 \frac{dP_i}{dP_j}dP_i=\int_{D} \rho_i\log_2 \frac{\rho_i}{\rho_j}dx
\end{equation}
The values of $\rho(t,\cdot), \rho^*(t,\cdot)$ 
are evaluated at the sequence of points $x_k$ 
generated by the finite volume method as described above, i.e. 
an approximation (more precisely, a piecewise-constant discretization) 
of the densities is needed, which is either observed or computed by trajectory observations
(see Section \ref{Sec:Results}).
We approximate the solution $\theta^*$ 
of (\ref{eq:density_learning}) with 
respect to 
$V_d(\theta) := \sum_{\tau=t_0}^{t_f} D_{KL}
    (P^*(\tau) || P(\tau))$,
with the iterative scheme
\begin{equation}
\theta^{n+1} = -\hat{\mathbf{H}}^{-1}(\theta^n)\nabla_\theta V_d(\theta^n)
\end{equation}
where $\hat{\mathbf{H}}$ is a positive-definite approximation of the Hessian computed via the Lanczos iteration \cite{nash}. The gradient is computed by the usual two-point finite difference formula. The KL divergence is approximated by a Riemann sum over the support of the observed density 
which is sampled over the same grid of points as the approximated density.

We note that in each iteration of the learning algorithm, 
the solution of the BVP associated with the system of PDEs (\ref{eq:euler}) must be numerically computed,
which has become feasible 
due to the computational advantages originating from 
the use of the proposed linear operator $\mathcal{L}_x$ in (\ref{eq:Lx}) 
(see Remarks \ref{rmk:1D}, \ref{rmk:2D}).
%
%


\subsection{Learning the Interaction Function from Particle Trajectories}
\label{sSec:particle_learning}

In order to better understand the computational advantages of the proposed 
methodology, we compare it here with a standard 
learning approach using trajectory data of the 
position and velocity $\cbra{(x_i,v_i)}_{i=1}^N$ of each particle for some large 
but finite number of particles $N$.
In general, this problem is a nonlinear system identification problem
with known system form given by (\ref{eq:cs}) and unknown 
interaction function ${\psi:\mathbb{R}^d \times \mathbb{R}^d \rightarrow \mathbb{R}}$.

Because of the nonlinearity of (\ref{eq:cs}) and the dependence of the right hand side
on every pair $\pbra{(v_i,v_j)}$ and $\pbra{(x_i,x_j)}$, $i\neq j$, 
system identification requires the solution of an ODE-constrained optimization problem
of dimension $O(N^2)$, which has a complexity of $O(N^2 N_s^2)$.
As a result, it is apparent that for large number of particles $N\gg N_s$,
the proposed mean-field methodology is quite faster (see Remarks \ref{rmk:1D}, \ref{rmk:2D}).

An energy-based approach is given in \cite{Matei2019}, where it is shown that 
the Cucker-Smale model (\ref{eq:cs})
is equivalent to a fully connected \textit{N}-dimensional network of generalized mass-spring-dampers with appropriately defined Hamiltonian functions, 
that can be written in an 
input-state-output port-Hamiltonian form \cite{h2}: 
\begin{equation}
\dot{z}=[J(z) - R(z)] \pder{H(z)}{z}
\label{eq:PH}
\end{equation}
where $z=(q,p)$, with $q,p\in\mathbb{R}^{\frac{N(N-1)}{2}}$ being 
the vectors of relative distances and momenta between 
each pair of particles, and 
the quantities $J=-J^T$, $H$ and $R$ are appropriately defined. 
The dependence of (\ref{eq:PH}) on the interaction function 
$\psi$ is introduced by the resistive term $R=R(\psi)$ \cite{Matei2019}, 
and is modeled as an artificial neural network with a single hidden layer. 
The parameters are represented by a vector $\theta$ and the 
learning process is formulated as a least-squares 
optimization problem 
\begin{equation}
    \min_{\theta}
    \sum_{\tau=t_0}^{t_f} 
    \| \dot z^*(\tau) - \dot z(\tau) \|^2
    \label{eq:particle_learning}
\end{equation}
where $z^*$ represent the observed trajectories, 
and $z$ are subject to (\ref{eq:PH}), and 
the solution $\theta^*$ 
of (\ref{eq:particle_learning}) with 
respect to $V_p(\theta) :=     \sum_{\tau=t_0}^{t_f} 
    \| \dot z^*(\tau) - \dot z(\tau) \|^2$, is approached
by an iterative gradient descent method
\begin{equation}
\theta^{n+1} = \theta^n - \alpha_n( \nabla_\theta V_p(\theta^n) ),~
n=0,1,2,\ldots 
\end{equation}
where the iteration maps 
$\alpha_n:\mathbb{R}^2 \rightarrow \mathbb{R}^2$, ${n\geq 0}$
are defined in accordance with the Adam method of moments
for stochastic optimization \cite{kingma2014adam}, 
and the computation of the gradient vectors is implemented 
using automatic differentiation \cite{autograd}. 
%
%
It is clear that the dimension of the dynamical system to be solved within the optimization problem grows quadratically with the number of particles $N$, 
affecting the scalability of such approaches.
Moreover, the quality of the observed trajectory data is crucial for the 
performance of the learning algorithm.


We note, however, that
there is a potential advantage in using the proposed interaction function model 
(\ref{eq:psi}), (\ref{eq:psi_coef}),
even in learning the interaction dynamics of a swarm 
from particle trajectories. 
First, the number of parameters to be estimated is greatly reduced, 
compared to a general regression function such as
a neural network \cite{Matei2019}, 
or a mixture of Gaussians
\cite{mao2019nonlocal}, 
which reduces the amount of data required for convergence.
Secondly, every update in the optimization algorithm 
improves the estimate of the interaction function over the entire
domain $D$, and not only over a small subset $D_o\subset D$
where the distances between each pair of interacting particles 
happen to be observed. 
This can result in faster, and, more importantly, robust estimation 
of the interaction function.

Finally, as an alternative to solving (\ref{eq:particle_learning}), 
in case observations of the particle trajectories are available, 
we can always numerically integrate to approximate $\rho$ and use this approximation
in our density-based learning algorithm. 
We will follow this approach when comparing the experimental results 
in the one-dimensional case in Section \ref{Sec:Results}.

\section{Numerical Results}
\label{Sec:Results}

\subsection{One-dimensional Case}

We illustrate our results 
in the domain 
$D=[-\pi,\pi]$ ($L=2\pi$), with 
initial density and bulk velocity given by
\begin{align}
\rho_0(\hat x) &= \frac{\pi}{2L} \cos\frac{\pi \hat x}{L}, \\ 
u_0(\hat x) &= - \sin\frac{\pi \hat x}{L},\ \hat x\in D,\
c>0
\label{eq:initial_conditions_1D} 
\end{align}
i.e. assuming that $\rho_0(\hat x) = u_0(\hat x) = 0,\ \forall \hat x\notin D$, where $\hat x$ is as defined in (\ref{eq:fluctuations}). 
In order to accurately evaluate the learning scheme defined in
Section \ref{Sec:density_learning}, 
we obtain the empirical density evolution data $\rho*$ 
by first simulating the particle equations (\ref{eq:cs})
with initial conditions randomly generated from the 
initial density and bulk velocity 
(\ref{eq:initial_conditions_1D}), 
and then taking the piecewise-constant density discretization
\begin{equation}
    \rho^*[t_i,x_s]: = \frac{1}{N_s\lambda(I_j)}\mu(\{ x_k([t_i] \in I_j\})
\end{equation}
where $\lambda(\cdot)$ is the Lebesgue measure, $\mu(\cdot)$ is the counting measure, and $I_i$, $x_j$ are defined as in the formulation of the finite volume method (Section \ref{sSec:ComputationalMethods}).
To showcase the robustness of our approach to noisy observations, 
we add a Gaussian noise $\epsilon_n \sim N(0,\sigma_n^2)$
with $\sigma_n^2=1$ to the trajectory data.
We choose an interaction function $\psi^*$ 
of the form (\ref{eq:psi}), (\ref{eq:psi_coef})
with $(\hat k^*,\hat \lambda^*)=(4,1)$. 
%
%
The system of particle equations  
is numerically solved using the velocity Verlet algorithm \cite{mao2019nonlocal},
which, given 
a system of ODEs of the form
\begin{equation}
    \begin{cases}
    \frac{dx}{dt} &= v \\
    \frac{dv}{dt} &= a(x,v,t),
    \end{cases}
\end{equation}
with appropriate initial conditions 
and a time-discretization at steps $\{0,1,...,i,...\}$ with increment $\Delta t$, takes the form
\begin{equation}
    \begin{split}
    v_{i+\frac{1}{2}} &= v_i + \frac{1}{2}a(x_i,v_i, t_i)\Delta t\\
    x_{i+1} &= x_i + \Delta t v_{i+\frac{1}{2}} \\
    v_{i+1} &= v_i + \frac{\Delta t}{2}[a(x_i,v_i,t_i) + a(x_{i+1}, v_{i+\frac{1}{2}}, t_{i+1})].
    \end{split}
\end{equation}
The agreement between
the solutions of the particle model (\ref{eq:cs}) and 
the macro-scale model (\ref{eq:pdes}) for 
$N_s=2 \cdot 10^4$, $\Delta t = .01$, and cell $\Delta x = \frac{2\pi}{101}$,
and  is shown in 
Fig. \ref{fig:macros_micro}.
The training error 
and the reconstructed interaction function are depicted 
in Fig. \ref{fig:training_ours_1D}.
The parameters $(\hat k^*,\hat \lambda^*)=(3.98721701,0.98546559) \sim (4,1)$ 
of the interaction function
$\psi$ were recovered and the Newton's iteration converged in 11 iterations.

\begin{figure}[ht!]
        \centering
        \includegraphics[trim=20 360 40 80, clip,width=0.4\textwidth]{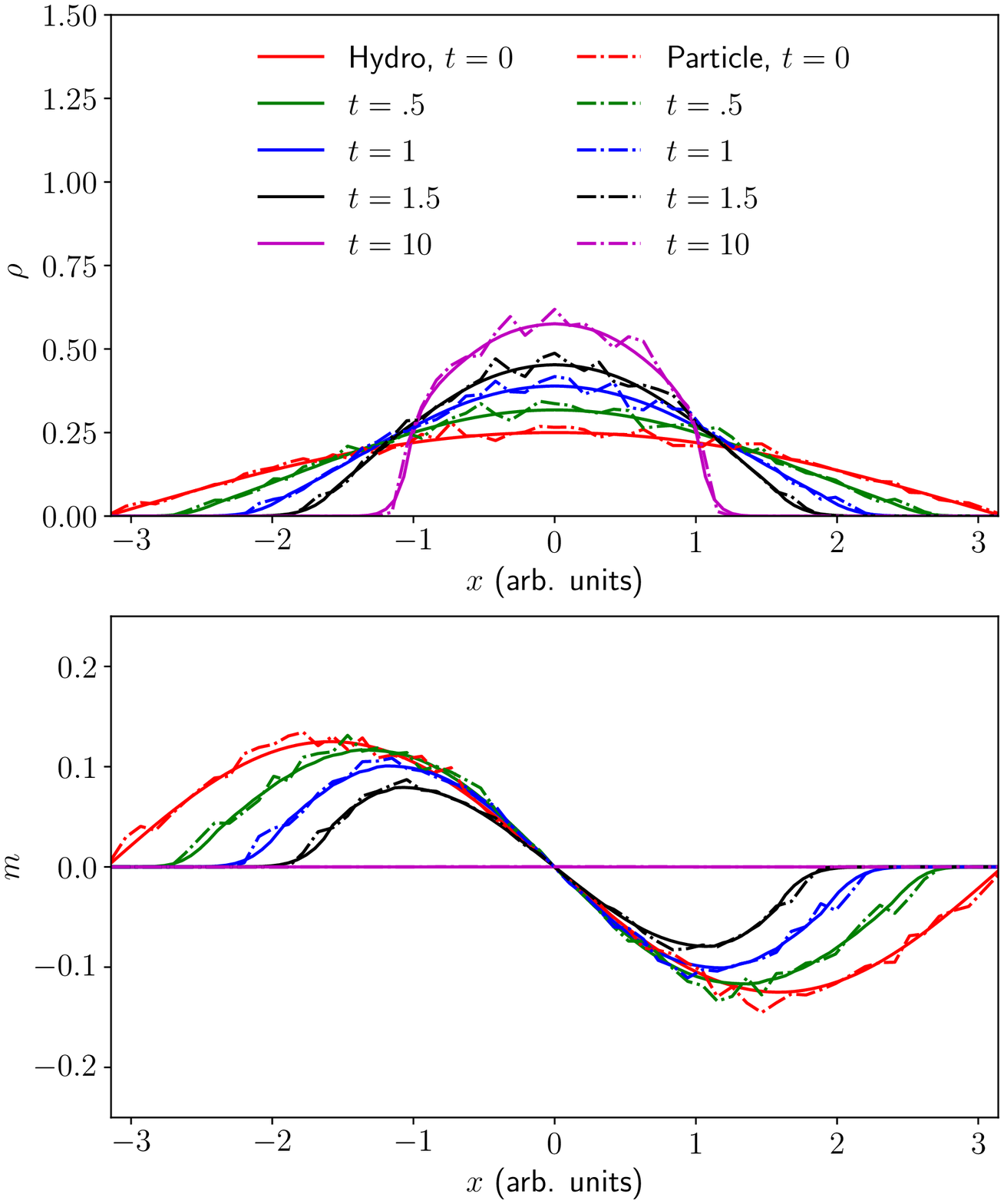}
        \vspace{-.5em}
        \caption{Evolution of the one-dimensional densities $\rho(t,\hat x)$ 
        as computed by solving the macro-scale model and the particle model (dashed-line).}
        \label{fig:macros_micro}
\end{figure}

\begin{figure}[h]
    \centering
    \includegraphics[trim=15 5 50 45, 
    clip,width=0.24\textwidth]{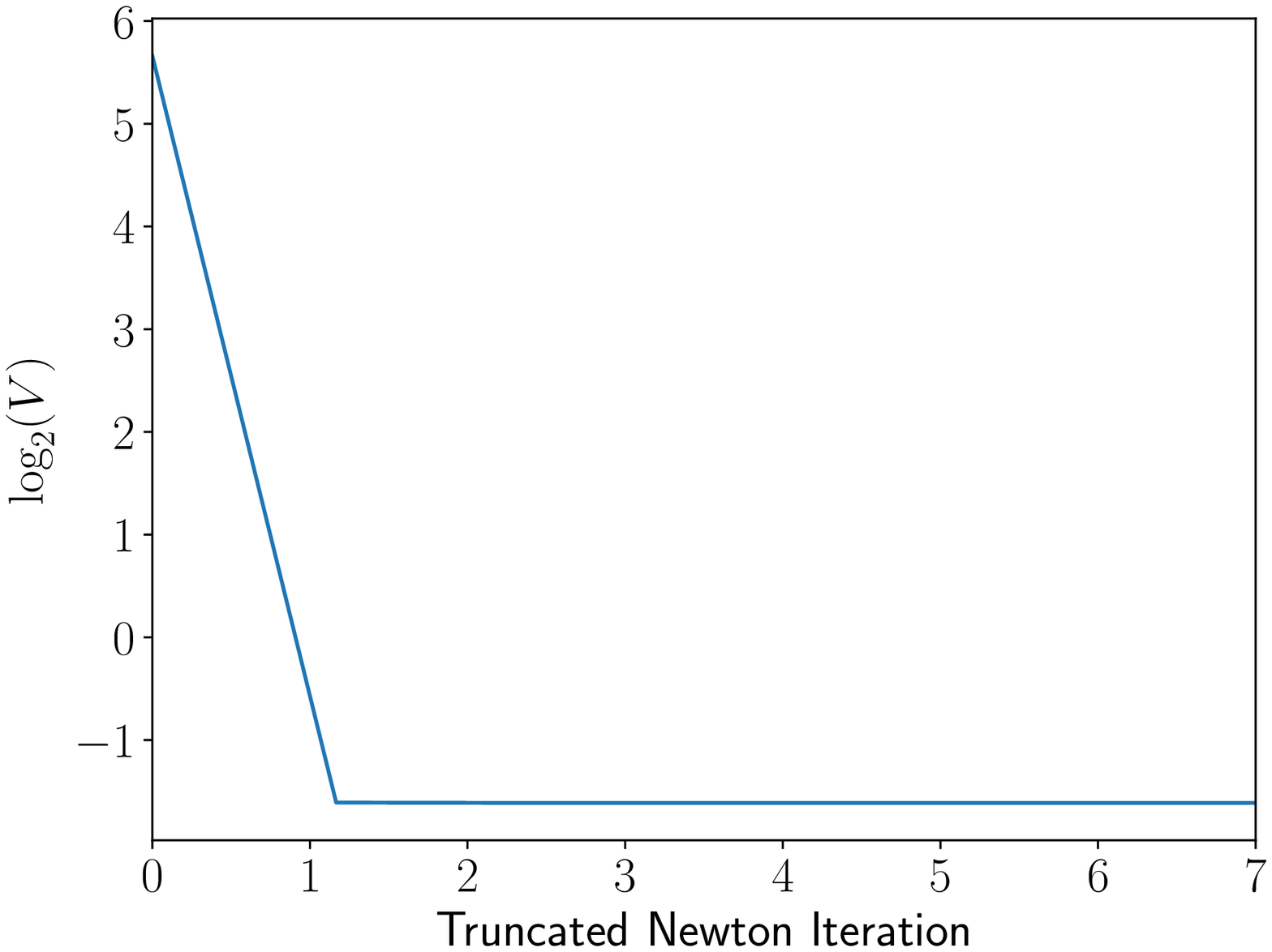}
    \includegraphics[trim=15 5 55 45, 
    clip,width=0.24\textwidth]{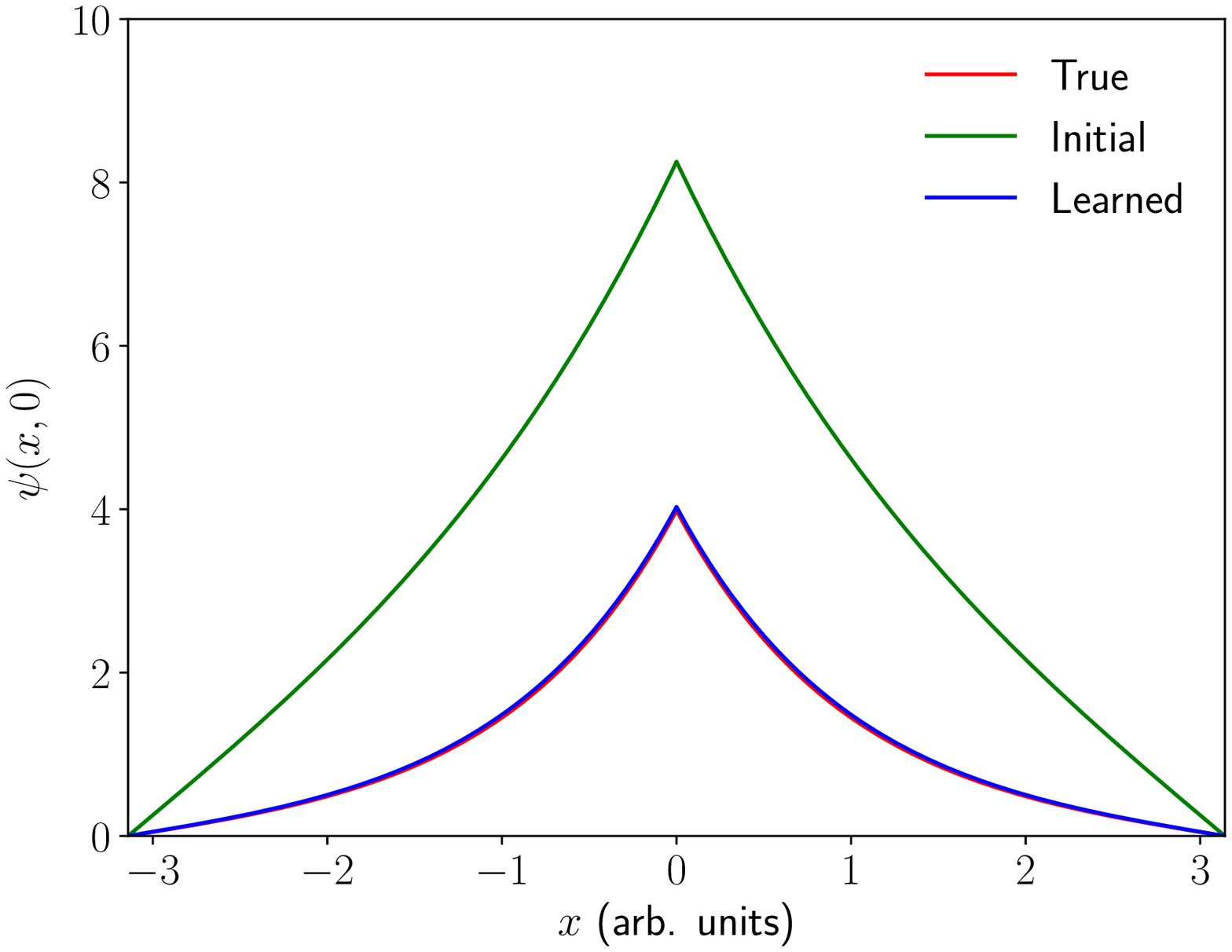}
    \caption{(left) 
    Training error 
    for the one-dimensional learning algorithm using
    observations of density evolution data. 
    (right) Reconstruction of the interaction function $\psi$.  
    Observed data generated by simulating the Cucker-Smale model (\ref{eq:cs})
    with the proposed interaction function $\psi^*$ as in 
    (\ref{eq:psi}), (\ref{eq:psi_coef})
    with $( k^*, \lambda^*)=(4,1)$. 
    }
    \label{fig:training_ours_1D}
\end{figure}

\begin{figure}[h]
    \centering
    \includegraphics[trim=15 0 40 45, 
    clip,width=0.24\textwidth]{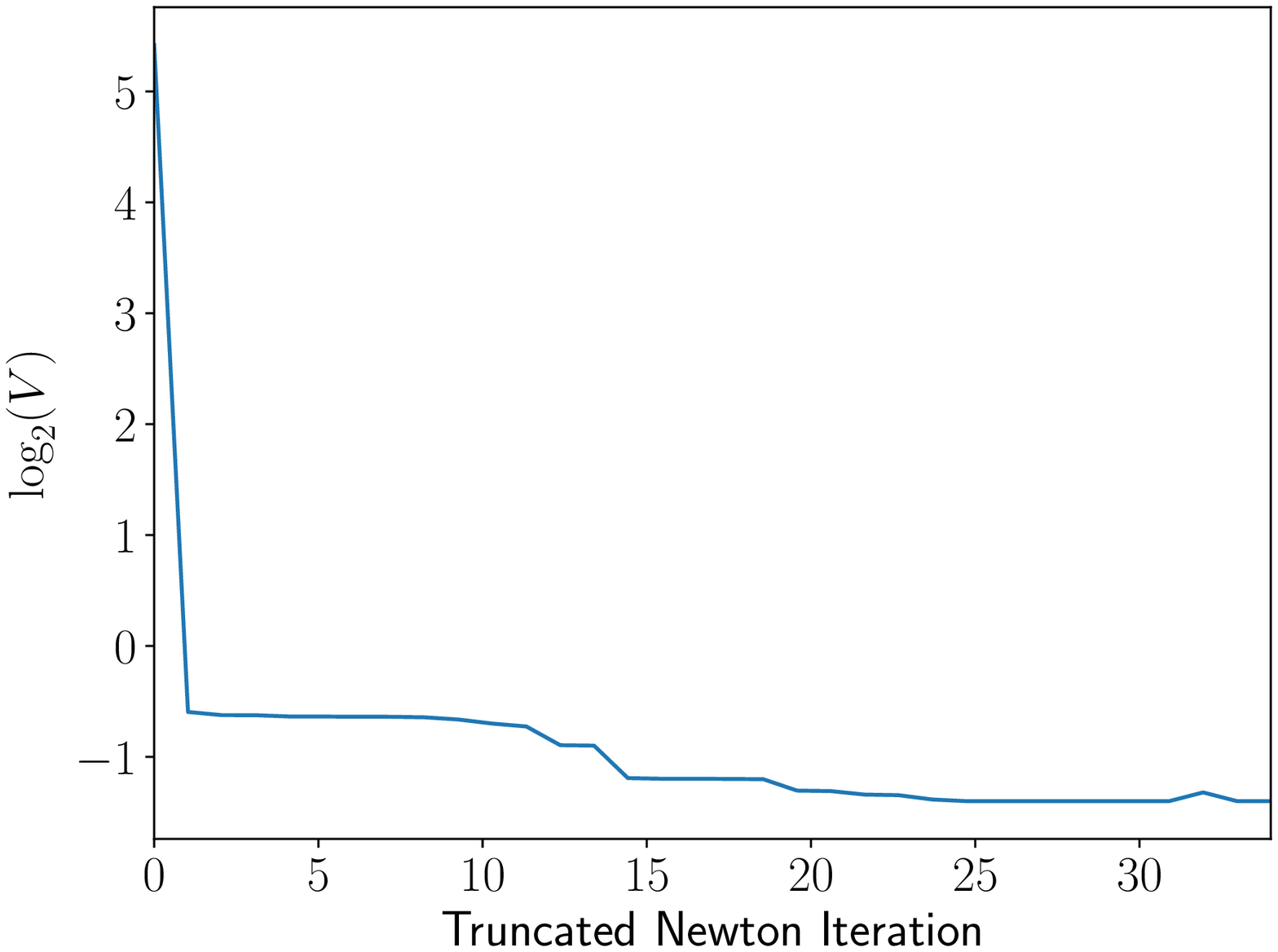}
    \includegraphics[trim=15 5 55 45, 
    clip,width=0.24\textwidth]{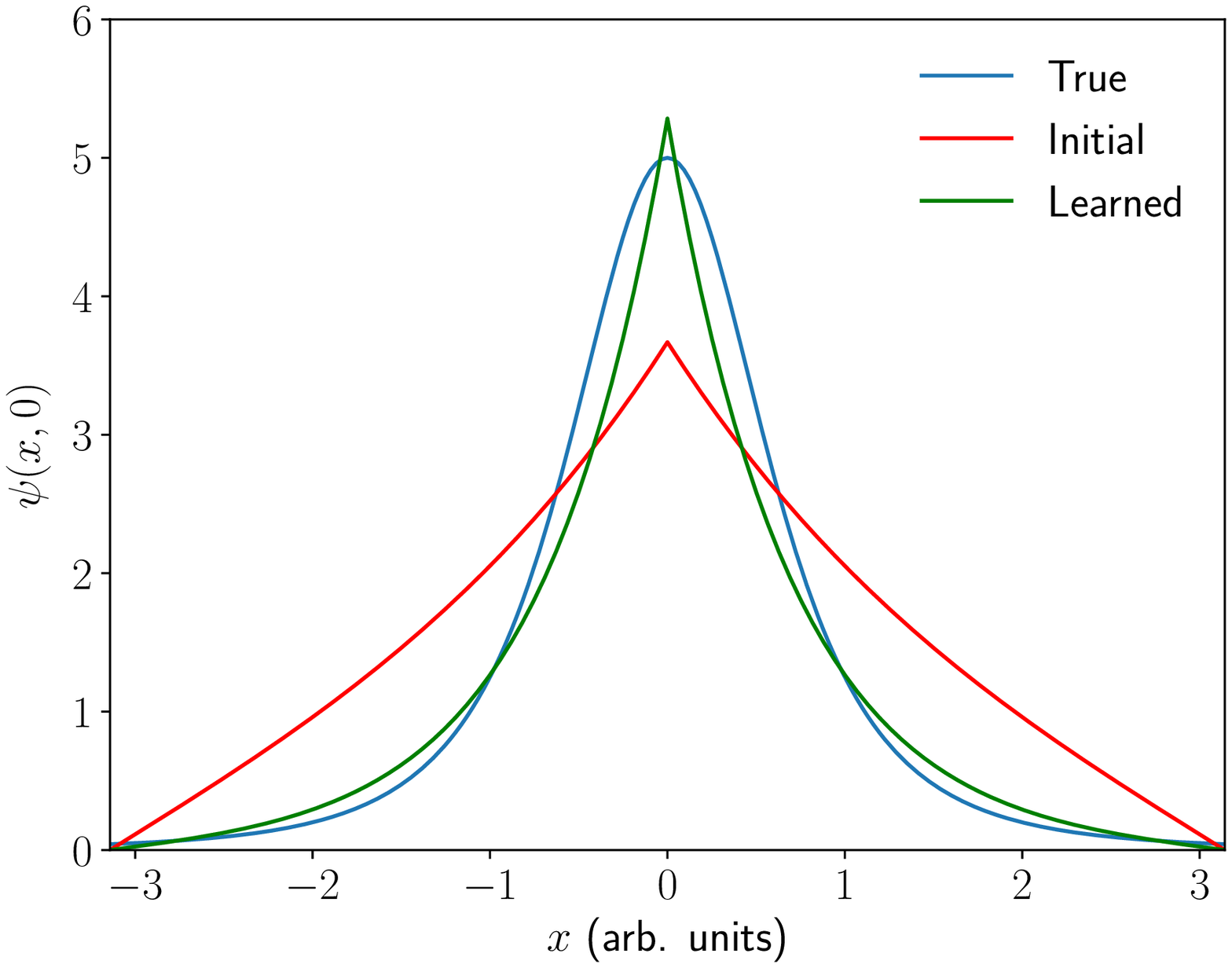}
    \caption{(left) 
    Training error 
    for the one-dimensional learning algorithm using
    observations of density evolution data. 
    (right) Reconstruction of the interaction function $\psi$.  
    Observed data generated by simulating the Cucker-Smale model (\ref{eq:cs})
    with the original interaction function $\psi^*$ in (\ref{eq:cs_psi_original}),
    with $( K^*, \gamma^*)=(5,2)$.
    }
    \label{fig:training_cs_1D}
\end{figure}

As a second experiment, in order to illustrate the expressiveness 
of the proposed family of interaction functions
(\ref{eq:psi}), (\ref{eq:psi_coef}), and assess the generalizability of the proposed
methodology, we obtain the ``observed'' density evolution 
by simulating system (\ref{eq:cs}) with the original Cucker-Smale interaction function 
\begin{equation}
\psi_{CS}(x,y) = \psi^*(x-y) = \frac{K^*}{(1+\|x-y\|^2)^{\gamma^*}}    
\label{eq:cs_psi_original}
\end{equation}
for $( K^*, \gamma^*)=(5,2)$.

The training error 
and the reconstructed interaction function are depicted 
in Fig. \ref{fig:training_cs_1D}.
We observe that the reconstruction is not ideal but closely resembles 
the original interaction function, while the reconstruction error of the density
evolution of the swarm is negligible. 
These results validate our hypothesis that the proposed interaction functions 
can model a wide range of collective behaviors, mostly because the 
model parameters can control the pairwise communication
of the swarm's agents without affecting the flocking behavior.



We note that problem (\ref{eq:density_learning}) is generally a 
non-convex optimization problem, and may 
be sensitive to initial estimates of the parameters 
$(k,\lambda)$ leading to sub-optimal solutions
$(\hat k^*,\hat \lambda^*)\neq (k^*,\lambda^*)$. 
In addition, the discretized objective function for the densities $V_d$
may approach very small values although
$(\hat k^*,\hat \lambda^*)\neq (k^*,\lambda^*)$, 
suggesting that, for a given set of observation data, 
certain non-global minima of 
(\ref{eq:density_learning}) can 
produce an accepted solution 
for the underlying interaction function of the swarm.
%
%
In this case, 
the reconstructed parameters $(\hat k^*,\hat \lambda^*)$ 
can be used 
to accurately reconstruct the 
actual observed trajectories.
%
%

\subsection{Two-dimensional Case}
\label{sSec:results2D}

We illustrate our results 
in the domain 
$D = [-\pi,\pi]\times[-\pi,\pi]$,
i.e. for $L=2\pi$, with 
initial density and bulk velocity given by
\begin{align}
\rho_0(\hat x, \hat y) &= \frac{\pi^2}{4L^2} \cos\frac{\pi \hat x}{L}\cos\frac{\pi \hat y}{L}, \\ 
u_0(\hat x, \hat y) &= -\frac{1}{4} (\sin\frac{\pi \hat x}{L},  \sin\frac{\pi \hat y}{L})^T, \hat x, \hat y \in [-\pi,\pi]
\label{eq:initial_conditions_2D} 
\end{align}
i.e. assuming that $\rho_0(\hat x) = u_0(\hat x) = 0,\ \forall \hat x\notin D$.
We note that these initial conditions and compact domain, again refer to 
the fluctuation variables $\hat x$ defined in (\ref{eq:fluctuations}).

In the two-dimensional case, 
we obtain the density data observations by directly solving the 
mean-field equations (\ref{eq:euler}) for two different Cucker-Smale models.
First we solve (\ref{eq:pdes2}) with the operator $\mathcal{L}_x$
as defined in (\ref{eq:Lx})
for $(\hat k^*,\hat \lambda^*)=(4,1)$.
An illustration of the density $\rho^*$ and momentum $m^*$ evolution over time is given 
in Fig. \ref{fig:density_contours_2D}. 
The training error for our learning scheme is depicted 
in Fig. \ref{fig:training_ours_2D}.
The parameters $k,\lambda$ were estimated as $(\hat k^*,\hat \lambda^*)=(4.01514
, 1.00194)$.

\begin{figure}[h]
    \centering
    \includegraphics[trim=5 0 35 40, 
    clip,width=0.24\textwidth]{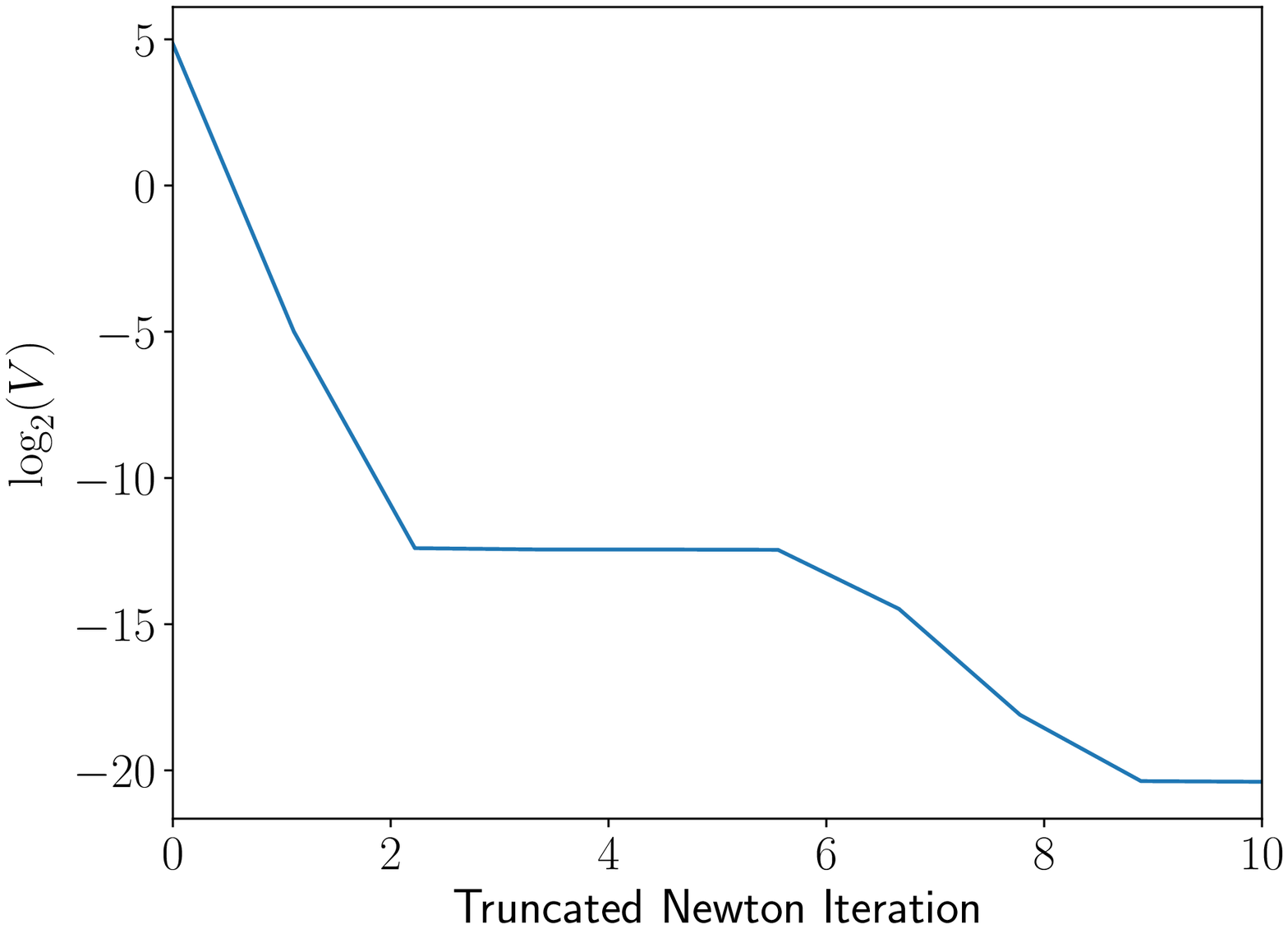}
    \includegraphics[trim=5 0 35 40, 
    clip,width=0.24\textwidth]{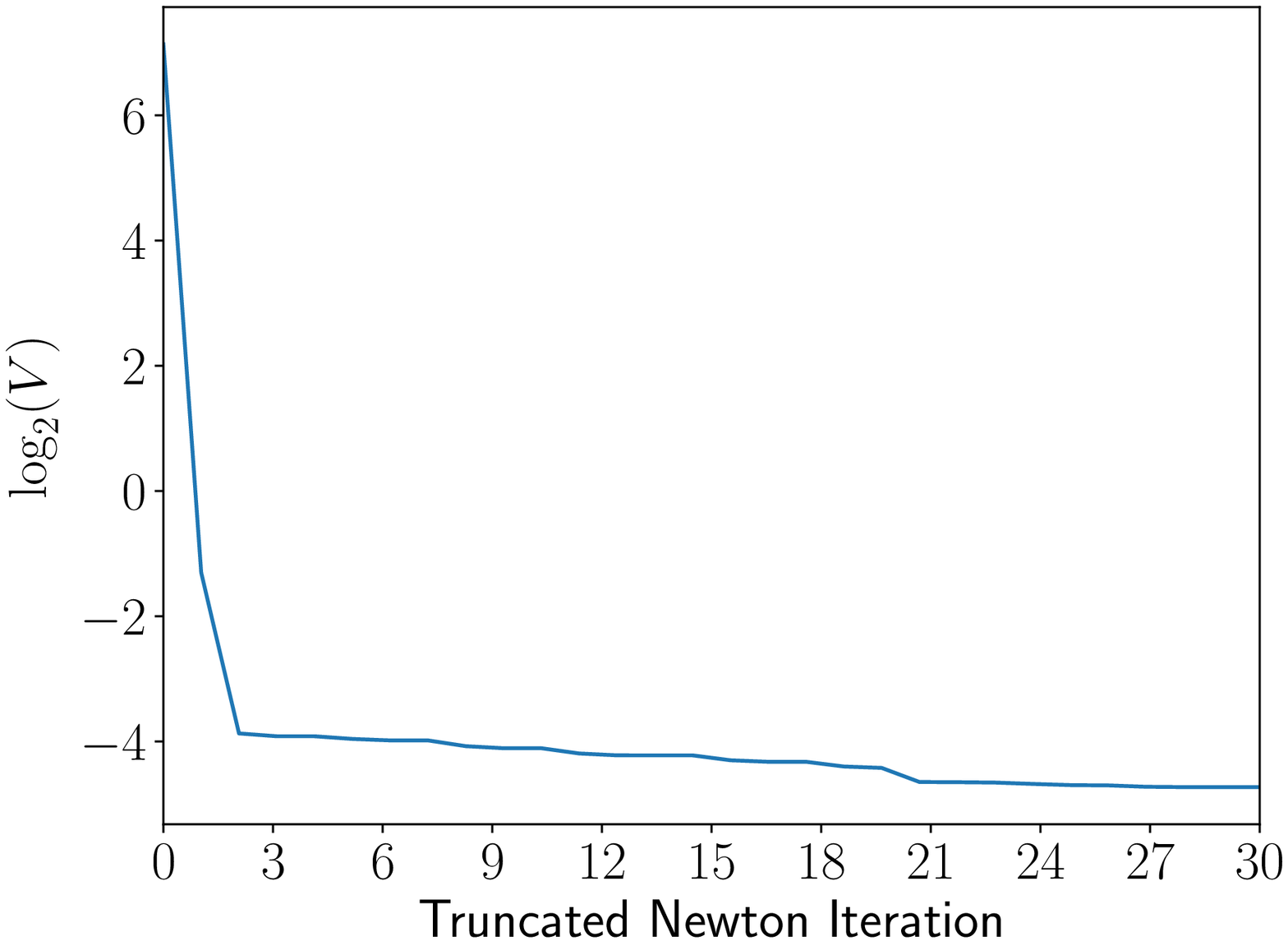}
    \caption{%
    Training error 
    for the two-dimensional learning algorithm. 
    (left)
    Observed density data generated by simulating the system of PDEs (\ref{eq:pdes2}), 
    with the operator $\mathcal{L}_x$ in (\ref{eq:Lx}) and $(k^*, \lambda^*) = (4, 1)$.
    (right)
    Observed density data generated by simulating the 2D Euler equations (\ref{eq:euler})
    with the original Cucker-Smale interaction function $\psi^*$ in 
    (\ref{eq:cs_psi_original}) for $( K^*, \gamma^*)=(5,2)$.
    Cost $V$ is plotted in $\log_2$ scale.
}
    \label{fig:training_ours_2D}
\end{figure}

Similar to the one-dimensional case, we test the 
generalizability of the proposed methodology, 
by obtaining the density evolution observations
by directly solving the Eulerian equations (\ref{eq:euler})
with the original Cucker-Smale interaction function $\psi^*$ in 
(\ref{eq:cs_psi_original}) for $( K^*, \gamma^*)=(5,2)$.


For the integral parts of (\ref{eq:euler}) of the form
\begin{equation}
    \phi(x,t) = \int_{D}\psi^*(x-s)q(s,t)ds = \int_{\mathbb R^2}\psi^*(x-s)q(s,t)ds,
\end{equation}
where $\psi^*$ are square-integrable kernels and $q$ compactly supported on square region $D$, we employ Fourier transform-based convolution. Via the properties of the convolution \cite{evans1998partial}, and denoting $\hat \phi(\xi, t), \hat \psi(\xi), \hat q(\xi,t)$ as the Fourier coefficients of the given functions, we apply the Fourier transform in $\mathbb {R}^2$, which gives:
\begin{equation}
    \hat \phi(\xi, t) = \hat \psi(\xi) \hat q(\xi,t).
\end{equation}
To implement this formula, we use the usual $2N_s$ zero-padded FFT on the regularly spaced points given in the hyperbolic solver to compute the coefficients and approximate the convolution integrals. This prevents circular convolution.

The training error is depicted 
in Fig. \ref{fig:training_ours_2D}.
Similar to the results  
in Fig. \ref{fig:training_cs_1D},
we expect that the reconstruction of the interaction function 
may not be ideal but can closely approximate 
the original interaction function, while the reconstruction error of the density
evolution of the swarm gets minimized. 
These results validate our hypothesis that the proposed interaction functions 
can model a wide range of collective behaviors in multi-dimensional space.


\section{Conclusion and Discussion}
\label{Sec:Conclusion}

We have considered the problem of understanding 
the coordinated movements of biological or artificial swarms.
While current
learning methodologies
mainly use agent-based models, accurate observations 
of the position and velocity trajectories
of each agent are required.  
Because of the difficulty to extract such observations in real life, 
we have proposed a learning scheme to reconstruct the 
coordination laws of the interacting agents 
from observations of the swarm's density evolution over time. 
%
We believe that developing learning algorithms based on the
macroscopic quantities of the swarm can play an important role 
in the analysis of collective motion
and has mainly been inhibited due to the computational expense
of solving the corresponding mean-field hydrodynamic equations.
%
%
%
%
%
The results of this work can 
be used to model and understand biological and artificial flocks, 
and design controllers for large networked systems and robotic swarms.
Moreover, the identification of the coordination laws of an observed swarm
through its density evolution over time, 
can lead to the development of fast defensive mechanisms against 
adversarial swarm attacks.

\bibliographystyle{IEEEtran} %
\bibliography{bib_pde_flocking,bib_intro,bib_sample}

\appendices

\section{Analytic Computation of the 1D Green's Function $\psi$} 
\label{App:psi}

\noindent

The Green's function $\psi(x,s)$ of the following BVP:
\begin{equation*}
\begin{cases}
-\frac{1}{2k}( y'' - \lambda^2 y) = f(x) \\
y(-\frac{L}{2}) = y(\frac{L}{2}) = 0 \\
-\frac{L}{2} \leq x \leq \frac{L}{2}
\end{cases}
\end{equation*}
takes the form:
\begin{equation*}
\psi(x,s) = \begin{cases}
a(s) e^{-\lambda x} + b(s) e^{\lambda x},~ x<s \\
c(s) e^{-\lambda x} + d(s) e^{\lambda x},~ x>s
\end{cases}
\end{equation*}
The first condition that $\psi(x,s)$ has to satisfy is 
$\psi(-\frac{L}{2},s)=0$, which gives:
\begin{equation}
b(s) = - a(s) e^{\lambda L}
\end{equation}
The second condition is $\psi(\frac{L}{2},s)=0$, which gives:
\begin{equation}
d(s) = - c(s) e^{-\lambda L}
\end{equation}
The third condition comes from the continuity of $\psi(x,s)$ at $x=s$:
\begin{equation}
a(s) (e^{-\lambda s} - e^{\lambda L} e^{\lambda s}) = 
c(s) (e^{-\lambda s} - e^{-\lambda L} e^{\lambda s})
\label{eq:cts}
\end{equation}
and the fourth is the differentiability condition at $x=s$:
\begin{equation}
a(s) (e^{-\lambda s} + e^{\lambda L} e^{\lambda s}) = 
c(s) (e^{-\lambda s} + e^{-\lambda L} e^{\lambda s}) - \frac{2k}{\lambda}
\label{eq:diff}
\end{equation}
Adding (\ref{eq:cts}) and (\ref{eq:diff}) gives:
\begin{equation*}
c(s) = a(s) + \frac{k}{\lambda} e^{\lambda s}
\end{equation*}
and, in addition, subtracting (\ref{eq:diff}) from (\ref{eq:cts}) gives:
\begin{equation*}
\begin{aligned}
a(s) &= K (e^{-\lambda s} -  e^{\lambda s} e^{-\lambda L}) \\
c(s) &= K (e^{-\lambda s} -  e^{\lambda s} e^{\lambda L})
\end{aligned}
\end{equation*}
where 
\begin{equation*}
K = -\frac{k}{\lambda} \frac{1}{e^{\lambda L}- e^{-\lambda L}}
\end{equation*}
Therefore, the Green's function $\psi(x,s)$ takes the form 
\begin{equation*}
\psi(x,s) = \begin{cases}
K (e^{-\lambda s} -  e^{\lambda s} e^{-\lambda L})
(e^{-\lambda x} -  e^{\lambda x} e^{\lambda L}),~ x<s \\
K (e^{-\lambda s} -  e^{\lambda s} e^{\lambda L})
(e^{-\lambda x} -  e^{\lambda x} e^{-\lambda L}),~ x>s
\end{cases}
\end{equation*}
which can be equivalently written 
(by multiplying by $e^{\lambda \frac L 2}e^{-\lambda \frac L 2}$) as 
\begin{equation*}
\psi(x,s) = \begin{cases}
K \sigma_m(s) \sigma_p(x),~ x<s \\
K \sigma_p(s) \sigma_m(x),~ x>s
\end{cases}
\end{equation*}
where 
\begin{equation*}
\begin{aligned}
\sigma_m(z) = 2 \sinh \pbra{\lambda(z-\frac L 2)},\
\sigma_p(z) = 2 \sinh \pbra{\lambda(z+\frac L 2)}
\end{aligned}
\end{equation*}
As a final note, it is clear that $\psi(x,s)$ satisfies the symmetry condition:
\begin{equation*}
\psi(x,s) = \psi(s,x).
\end{equation*}


\begin{IEEEbiography}[{\includegraphics[width=1in,height=1.25in,clip,keepaspectratio]
{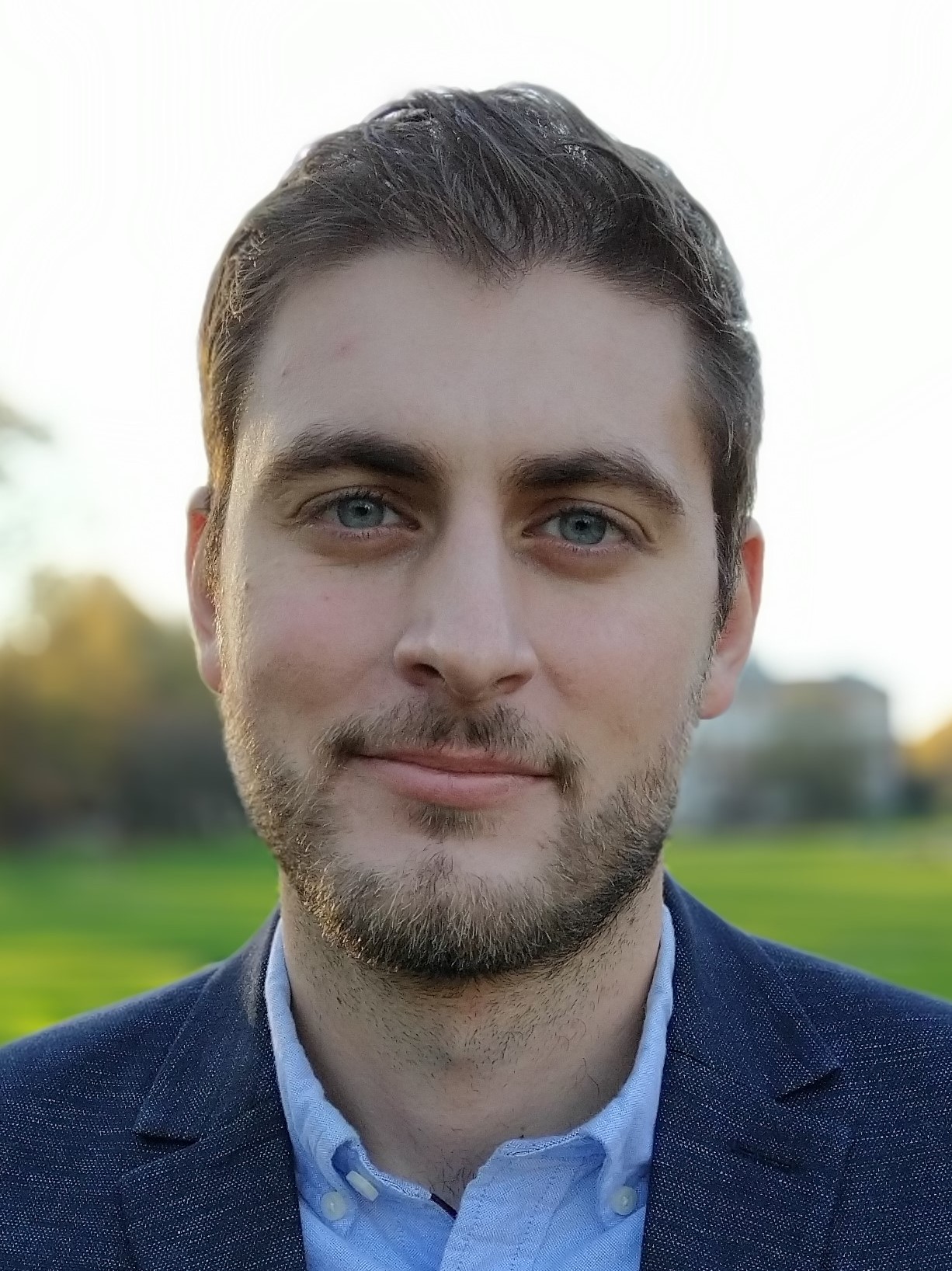}}]{Christos N. Mavridis} (M'20) 
received the Diploma degree in electrical and computer engineering from the National Technical University of Athens, Greece, in 2017,
and the M.S. and  Ph.D. degrees in electrical and computer engineering at the University of Maryland, College Park, MD, USA, in 2021. 
His research interests include learning theory, stochastic optimization, systems and control theory, multi-agent systems, and robotics. 

He has worked as a researcher at the Department of Electrical and Computer Engineering at the University of Maryland, College Park, MD, USA, and as a research intern for the Math and Algorithms Research Group at Nokia Bell Labs, NJ, USA, and the System Sciences Lab at Xerox Palo Alto Research Center (PARC), CA, USA. 

Dr. Mavridis is an IEEE member, and a member of the Institute for Systems Research (ISR) and the Autonomy, Robotics and Cognition (ARC) Lab. He received the Ann G. Wylie Dissertation Fellowship in 2021, and the A. James Clark School of Engineering Distinguished Graduate Fellowship, Outstanding Graduate Research Assistant Award, and Future Faculty Fellowship, in 2017, 2020, and 2021, respectively. He has been a finalist in the Qualcomm Innovation Fellowship US, San Diego, CA, 2018, and he has received the Best Student Paper Award (1st place) in the IEEE International Conference on Intelligent Transportation Systems (ITSC), 2021.
\end{IEEEbiography}

\begin{IEEEbiography}[{\includegraphics[width=1in,height=1.25in,clip,keepaspectratio]
{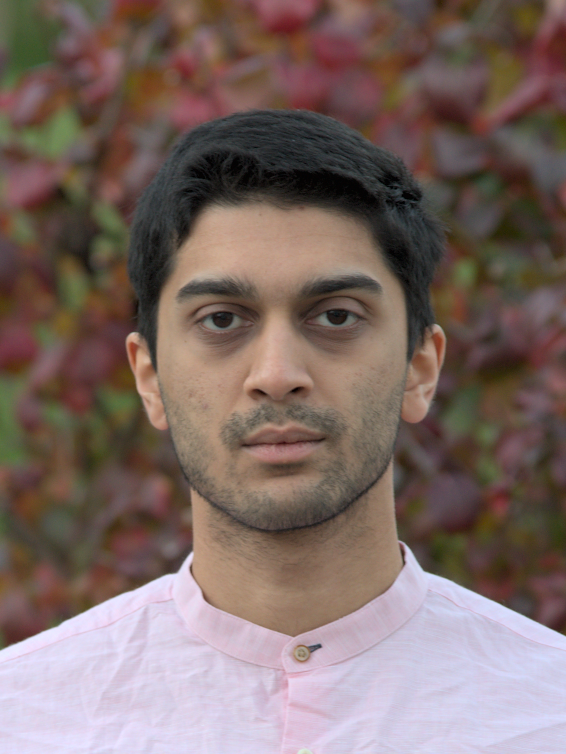}}]{Amoolya Tirumalai} (M'19) received the Bachelor of Science degree in 
biomedical engineering from the Georgia Institute of Technology, Atlanta, GA, USA, in 2018. 
He is currently pursuing his Ph.D. degree in electrical and computer engineering at the 
University of Maryland, College Park, MD, USA. His interests lie in
optimization and optimal control theory, distributed parameter systems, collective motion, and cyber-physical systems.

Beginning in 2019, Mr. Tirumalai has worked as a research assistant in the Department of Electrical and Computer Engineering at the University of Maryland. He previously was an Associate in Research at the Department of Biological Sciences at Duke University, Durham, NC, USA. 

Mr. Tirumalai is a member of the Institute for Systems Research. He is a 2019 recipient of the
Clark Doctoral Fellowship from the A. James and Alice B. Clark Foundation and the University of Maryland.
\end{IEEEbiography}

\begin{IEEEbiography}[{\includegraphics[width=1in,height=1.25in,clip,keepaspectratio]
{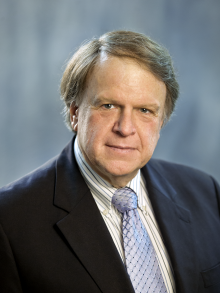}}]{John S. Baras} (F'13) 
received the Diploma degree in electrical and mechanical engineering from the National Technical University of Athens, Athens, Greece, in 1970, and the M.S. and Ph.D. degrees in applied mathematics from Harvard University, Cambridge, MA, USA, in 1971 and 1973, respectively.

He is a Distinguished University Professor and holds the Lockheed Martin Chair in Systems Engineering, with the Department of Electrical and Computer Engineering and the Institute for Systems Research (ISR), at the University of Maryland College Park. From 1985 to 1991, he was the Founding Director of the ISR. Since 1992, he has been the Director of the Maryland Center for Hybrid Networks (HYNET), which he co-founded. His research interests include systems and control, optimization, communication networks, applied mathematics, machine learning, artificial intelligence, signal processing, robotics, computing systems, security, trust, systems biology, healthcare systems, model-based systems engineering.

Dr. Baras is a Fellow of IEEE (Life), SIAM, AAAS, NAI, IFAC, AMS, AIAA, Member of the National Academy of Inventors and a Foreign Member of the Royal Swedish Academy of Engineering Sciences. Major honors include the 1980 George Axelby Award from the IEEE Control Systems Society, the 2006 Leonard Abraham Prize from the IEEE Communications Society, the 2017 IEEE Simon Ramo Medal, the 2017 AACC Richard E. Bellman Control Heritage Award, the 2018 AIAA Aerospace Communications Award. In 2016 he was inducted in the A. J. Clark School of Engineering Innovation Hall of Fame. In 2018 he was awarded a Doctorate Honoris Causa by his alma mater the National Technical University of Athens, Greece.   
\end{IEEEbiography}

\end{document}